\documentclass[a4paper,accepted=2023-12-09,onecolumn,11pt]{quantumarticle}
\pdfoutput=1

\usepackage{amsmath}
\usepackage{amssymb}
\usepackage{amsthm}
\usepackage[normalem]{ulem}
\usepackage{env}
\usepackage{physics}
\usepackage{cleveref}
\crefname{theorem}{Theorem}{Theorems}
\Crefname{theorem}{Theorem}{Theorems}
\Crefname{corollary}{Corollary}{Corollary}

\title{Quantum-inspired permanent identities}

\author{Ulysse Chabaud}
\orcid{0000-0003-0135-9819}
\affiliation{Institute for Quantum Information and Matter, California Institute of Technology, Pasadena, CA 91125, USA}
\email{uchabaud@caltech.edu}
\author{Abhinav Deshpande}
\orcid{0000-0002-6114-1830}
\affiliation{Institute for Quantum Information and Matter, California Institute of Technology, Pasadena, CA 91125, USA}
\author{Saeed Mehraban}
\orcid{0000-0002-1323-360X}
\affiliation{Computer Science, Tufts University, Medford, MA 02155, USA}
\email{saeed.mehraban@tufts.edu}

\begin{document}

\maketitle

\begin{abstract}
    The permanent is pivotal to both complexity theory and combinatorics. 
    In quantum computing, the permanent appears in the expression of output amplitudes of linear optical computations, such as in the Boson Sampling model.
    Taking advantage of this connection, we give quantum-inspired proofs of many existing as well as new remarkable permanent identities. Most notably, we give a quantum-inspired proof of the MacMahon master theorem as well as proofs for new generalizations of this theorem. Previous proofs of this theorem used completely different ideas.
    Beyond their purely combinatorial applications, our results demonstrate the classical hardness of exact and approximate sampling of linear optical quantum computations with input cat states. 
\end{abstract}

\tableofcontents

\section{Introduction}

The \textit{permanent} of an $m\times m$ matrix $A=(a_{ij})_{1\le i,j\le m}$ is a combinatorial function defined as~\cite{minc1984permanents,percus2012combinatorial}:
\begin{equation}
    \mathrm{Per}(A)=\sum_{\sigma\in\mathcal S_m}\prod_{k=1}^ma_{k\sigma(k)},
\end{equation}
where $\mathcal S_m$ is the symmetric group over $m$ symbols. 
While the closely related determinant can be computed efficiently using many methods such as Gaussian elimination, Valiant famously proved the $\#$\textsf P-hardness of computing the permanent~\cite{valiant1979complexity}.

Interestingly, the permanent appears in the expression of the output amplitudes of linear optical quantum computations with noninteracting bosons~\cite{caianiello1953quantum,scheel2004permanents}, as in the Boson Sampling model of quantum computation~\cite{Aaronson2013}. 
This connection has lead to several linear optical proofs of existing and new classical complexity results: computation of the permanent is $\#$\textsf P-hard~\cite{aaronson2011linear}, (inverse polynomial) multiplicative estimation of the permanent of positive semidefinite matrices is in $\textsf{BPP}^\textsf{NP}$~\cite{rahimi2015can}, multiplicative estimation of the permanent of orthogonal matrices is $\#$\textsf{P}-hard~\cite{grier2016new}, and computation of a class of multidimensional integrals is $\#$\textsf P-hard~\cite{rohde2016quantum}. It has also lead to the introduction of a quantum-inspired classical algorithm for additive estimation of the permanent of positive semidefinite matrices~\cite{chakhmakhchyan2017quantum}. Moreover, an approach inspired by quantum tomography recently showed that (subexponential) multiplicative estimation of the permanent of positive semidefinite matrices is \textsf{NP}-hard~\cite{meiburg2021inapproximability}.

Beyond its importance for complexity theory, the permanent has numerous applications for solving combinatorial problems~\cite{minc1984permanents,percus2012combinatorial} and identities for the permanent have been instrumental in these applications. For example, the MacMahon master theorem~\cite{macmahon2001combinatory}, which relates the permanent to a coefficient of the Taylor series of a determinant, is an invaluable tool for proving combinatorial identities~\cite{good1962proofs,carlitz1974application,carlitz1977some}. Similarly, Ryser’s formula~\cite{ryser1963combinatorial} and Glynn’s formula~\cite{balasubramanian1980combinatorics,bax1998finite,glynn2010permanent} are routinely used to compute the permanent more efficiently than the naive brute-force approach.

While linear optics has been used as a tool to explore the classical complexity of the permanent, previous work suggests that it can also be a useful way to obtain simple proofs of theorems about the permanent: for instance, \cite{scheel2004permanents} shows that the permanent of a unitary matrix $U$ lies in the (closed) complex unit disk by expressing $|\mathrm{Per}(U)|^2$ as an output probability of a linear optical sampling computation, while~\cite{Aaronson2013} derives simple permanent identities using various representations of the same linear optical sampling computation.
This begs the following questions:
\vspace{-0.1cm}
\begin{enumerate}
    \item Can we use linear optics to prove existing remarkable permanent identities, such as the MacMahon master theorem?
    \vspace{-0.1cm}
    \item Can we use linear optics to derive \textit{new} remarkable permanent identities?
\end{enumerate}
\vspace{-0.1cm}
\noindent In this work, we show that the answer to both questions is \textit{yes}.
\vspace{-0.1cm}
\medskip

We give quantum-inspired proofs of several important permanent identities in section~\ref{sec:qproofs}. In particular, we show that the MacMahon master theorem can be understood as two different ways of computing an inner product between two Gaussian quantum states.

We also derive new quantum-inspired identities for the permanent. Our main results are summarized in section~\ref{sec:results} and proven in section~\ref{sec:new}. These include generalizations of the MacMahon master theorem (Theorems~\ref{th:mmmt+} and~\ref{th:mmmt++}), new generating functions for the permanent (Theorem~\ref{th:gen}) and a formula for the sum of two permanents (Theorem~\ref{th:sumper}).

As a bonus, our findings also have consequences for the classical complexity of exact and approximate sampling of linear optical sampling computations with input cat states (Theorem~\ref{th:hardBScat}), which we discuss in section~\ref{sec:BScat}. We rigorously prove that the corresponding quantum probability distributions are as hard to sample as the original Boson Sampling distribution~\cite{Aaronson2013}, for all cat state amplitudes in the exact case and for small enough amplitudes in the approximate case. Until now, a formal proof was available only in the exact case~\cite{rohde2015evidence}.

\section{Background}
\label{sec:preliper}

\subsection{Notations and preliminary material}

\begin{table}[h]
    \centering
        \begin{tabular}{||c|c||}
    \hline
         \,&\,\\
         $\bm0=(0,\dots,0)$ & $\bm p\le\bm q\Leftrightarrow\forall k\in\{1,\dots,m\},\;p_k\le q_k$\\
         $\bm1=(1,\dots,1)$ & $-\bm z=(-z_1,\dots,-z_m)$\\
         $|\bm p|=p_1+\dots+p_m$ & $\bm z^*=(z_1^*,\dots,z_m^*)$\\
         $\bm p!=p_1!\dots p_m!$ & $\|\bm z\|^2=|z_1|^2+\dots+|z_m|^2$\\
         $\ket{\bm p}=\ket{p_1}\otimes\dots\otimes\ket{p_m}$ & $\bm z^{\bm p}=z_1^{p_1}\dots z_m^{p_m}$\\
         $\bm p+\bm q=(p_1+q_1,\dots,p_m+q_m)$ & $\partial_{\bm z}^{\bm p}=\partial_{z_1}^{p_1}\dots\partial_{z_m}^{p_m}$\\
         $\bm p\oplus\bm q=(p_1,\dots,p_m,q_1,\dots,q_m)$ & $d^m\bm zd^m\bm z^*=d\Re{z_1}d\Im{z_1}\dots d\Re{z_m}d\Im{z_m}$\!\\
         \,&\,\\
         \hline
        \end{tabular}
    \caption{Multi-index notations used in this paper, for $m\in\mathbb N^*$, $\bm p=(p_1,\dots,p_m)\in\mathbb N^m$, $\bm q=(q_1,\dots,q_m)\in\mathbb N^m$ and $\bm z=(z_1,\dots,z_m)\in\mathbb C^m$.}
    \label{tab:multi}
\end{table}

\noindent We use bold math for multi-index expressions (see Table~\ref{tab:multi} above). We denote by $\mathbb T=\{z\in\mathbb C,|z|=1\}$ the complex unit circle. For all $m\in\mathbb N^*$, all $\bm z=(z_1,\dots,z_m)\in\mathbb C^m$ and all $\bm p=(p_1,\dots,p_m)\in\mathbb N^m$, we use the notation $[\bm z^{\bm p}]$ to denote the coefficient of $\bm z^{\bm p}=z_1^{p_1}\dots z_m^{p_m}$ in an analytic expression. 
For all $m\in\mathbb N^*$, all $\bm p=(p_1,\dots,p_m)\in\mathbb N^m$ and all $\bm q=(q_1,\dots,q_m)\in\mathbb N^m$, we denote by $A_{\bm p,\bm q}$ the matrix obtained from an $m\times m$ matrix $A$ by first repeating its $i^{th}$ row $p_i$ times (deleting the row if $p_i=0$) for all $i\in\{1,\dots,m\}$ and then repeating its $j^{th}$ column $q_j$ times (deleting the column if $q_j=0$) for all $j\in\{1,\dots,m\}$. By convention, we set the permanent of non-square matrices to $0$ and $\mathrm{Per}(A_{\bm0,\bm0})=1$.

\medskip

Let $m\in\mathbb N^*$ denote the number of modes. We denote (unnormalized) quantum states using the Dirac ket notation. These are vectors in an infinite-dimensional Hilbert space spanned by the orthonormal Fock basis
\begin{equation}
    \{\ket{\bm p}=\ket{p_1}\otimes\dots\otimes\ket{p_m}\}_{\bm p=(p_1,\dots,p_m)\in\mathbb N^m}.
\end{equation}
Hereafter, we label Fock states using $p$, $q$, $\bm p$, $\bm q$, and $\bm k$, coherent states using $\alpha$, $\beta$, $\bm\alpha$, and $\bm\beta$, and two-mode squeezed states using $\lambda$, $\mu$, $\bm\lambda$, and $\bm\mu$. In particular, coherent states are defined as
\begin{equation}\label{eq:coh}
    \ket{\alpha}=e^{-\frac12|\alpha|^2}\sum_{p\ge0}\frac{\alpha^p}{\sqrt{p!}}\ket p,
\end{equation}
for all $\alpha\in\mathbb C$, and (unnormalized) two-mode squeezed states as
\begin{equation}
    \ket{\lambda}=\sum_{p\ge0}\lambda^p\ket{pp},
\end{equation}
for all $\lambda\in\mathbb C$ with $|\lambda|<1$. Moreover, cat states are defined as
\begin{equation}\label{eq:cat}
    \ket{\mathrm{cat}_\alpha}=\frac{e^{\frac12|\alpha|^2}}{2\sqrt{\sinh(|\alpha|^2)}}(\ket{\alpha}-\ket{-\alpha}),
\end{equation}
for all $\alpha\in\mathbb C$.

\medskip

We make use of the following inner products involving these states: for all $p\in\mathbb N$, $\alpha,\beta\in\mathbb C$, and $\lambda\in\mathbb C$ with $|\lambda|<1$,
\begin{equation}\label{eq:quantum1}
    \braket{p}{\alpha}=e^{-\frac12|\alpha|^2}\frac{\alpha^p}{\sqrt{p!}}
\end{equation}
\begin{equation}\label{eq:quantum1+}
    \braket{\alpha}{\beta}=e^{-\frac12|\alpha|^2-\frac12|\beta|^2+\alpha^*\beta},
\end{equation}
\begin{equation}\label{eq:quantum2}
    \braket{pp}{\lambda}=\lambda^p,
\end{equation}
\begin{equation}\label{eq:sqcoh}
    \bra{\lambda}(\ket\alpha\otimes\ket\beta)=e^{-\frac12|\alpha|^2-\frac12|\beta|^2+\lambda^*\alpha\beta},
\end{equation}
which can be readily checked in the Fock basis.
Moreover, for all $\bm p,\bm q\in\mathbb N^m$~\cite{scheel2004permanents,Aaronson2013},
\begin{equation}\label{eq:quantum3}
    \bra{\bm p}\hat U\ket{\bm q}=\frac{\mathrm{Per}(U_{\bm p,\bm q})}{\sqrt{\bm p!\bm q!}},
\end{equation}
where $\hat U$ is a passive linear operation over $m$ modes whose action on the creation operators of the modes is described by the unitary matrix $U=(u_{ij})_{1\le i,j\le m}$ as
\begin{equation}\label{eq:actioncrea}
    \hat U\hat a^\dag_j\hat U^\dag=\sum_{i=1}^mu_{ij}\hat a^\dag_i,
\end{equation}
for all $j\in\{1,\dots,m\}$.
As their name indicates, passive linear operations do not change the total number of photons~\cite{weedbrook2012gaussian}: for all $n\in\mathbb N$,
\begin{equation}\label{eq:quantum4}
    \hat U\Pi_n=\Pi_n\hat U,
\end{equation}
where $\Pi_n:=\sum_{|\bm p|=n}\ket{\bm p}\!\bra{\bm p}$ is the $m$-mode projector onto states with total photon number equal to $n$.
Passive linear operations map tensor products of coherent states to tensor products of coherent states~\cite{weedbrook2012gaussian}: for all $\bm\alpha\in\mathbb C^m$,
\begin{equation}\label{eq:quantum5}
    \hat U\ket{\bm \alpha}=\ket{U\bm\alpha}.
\end{equation}

\medskip

Recall that coherent states form an overcomplete basis:
\begin{equation}\label{eq:overcomp}
    \int_{\bm\alpha\in\mathbb C^m}\ket{\bm\alpha}\!\bra{\bm\alpha}\frac{d^m\bm\alpha d^m\bm\alpha^*}{\pi^m}=\hat{\mathbb I},
\end{equation}
where $\hat{\mathbb I}$ is the identity operator over $m$ modes.
We will also make use of the following Gaussian inner product: for all $\bm\lambda,\bm\mu\in\mathbb C^m$ with $|\lambda_k|<1$ and $|\mu_k|<1$ for all $k\in\{1,\dots,m\}$, and for all passive linear operation $\hat U$ over $2m$ modes,
\begin{equation}\label{eq:quantum6}
    \begin{aligned}
        \bra{\bm\lambda^*}\hat U\ket{\bm\mu}&=\int_{\bm\beta\in\mathbb C^{2m}}\bra{\bm\lambda^*}\ket{\bm\beta^*}\!\bra{\bm\beta^*}\hat U\ket{\bm\mu}\frac{d^{2m}\bm\beta d^{2m}\bm\beta^*}{\pi^{2m}}\\
        &=\int_{\bm\beta\in\mathbb C^{2m}}\bra{\bm\lambda^*}\ket{\bm\beta^*}\left(\bra{\bm\mu}{U^\dag\bm\beta^*}\rangle\right)^*\frac{d^{2m}\bm\beta d^{2m}\bm\beta^*}{\pi^{2m}}\\
        &=\int_{\bm\beta\in\mathbb C^{2m}}e^{\sum_{k=1}^m\lambda_k\beta_k^*\beta_{m+k}^*}e^{\sum_{k=1}^m\mu_k(U^T\bm\beta)_k(U^T\bm\beta)_{m+k}}e^{-\sum_{j=1}^{2m}\beta_j^*\beta_j}\frac{d^{2m}\bm\beta d^{2m}\bm\beta^*}{\pi^{2m}}\\
        &=\int_{\bm\beta\in\mathbb C^{2m}}\exp\left[-\frac12\begin{pmatrix}\bm\beta\\\bm\beta^*\end{pmatrix}^T\!\!\!V_U(\bm\lambda,\bm\mu)\begin{pmatrix}\bm\beta\\\bm\beta^*\end{pmatrix}\right]\frac{d^{2m}\bm\beta d^{2m}\bm\beta^*}{\pi^{2m}}\\
        &=\frac1{\sqrt{\mathrm{Det}(V_U(\bm\lambda,\bm\mu))}},
    \end{aligned}
\end{equation}
where we have used the overcompleteness of coherent states (\ref{eq:overcomp}) over $2m$ modes in the first line, the action of passive linear operations on coherent states (\ref{eq:quantum5}) in the second line, the overlap between two-mode squeezed states and coherent states (\ref{eq:sqcoh}) in the third line, and where we have introduced in the fourth line the $(4m)\times(4m)$ symmetric matrix
\begin{equation}\label{eq:defVU}
    V_U(\bm\lambda,\bm\mu):=\begin{pmatrix}-UV_{\bm\mu}U^T&I_{2m}\\I_{2m}&-V_{\bm\lambda}\end{pmatrix},
\end{equation}
where for all $\bm w\in\mathbb C^n$,
\begin{equation}\label{eq:Vw}
    V_{\bm w}:=\begin{pmatrix}0_m&\mathrm{Diag}(\bm w)\\\mathrm{Diag}(\bm w)&0_m\end{pmatrix}.
\end{equation}
Note that we associate mode $1$ with mode $m+1$, mode $2$ with mode $m+2$, and so on, i.e.\ $\ket{\bm\lambda}=\bigotimes_{k=1}^m\ket{\lambda_k}$, where $\ket{\lambda_k}$ is an unormalized two-mode squeezed state over modes $k$ and $m+k$.

\medskip

Finally, we note that any $m\times m$ matrix $B$ with $\|B\|\le1$, where $\|\cdot\|$ is the spectral norm, may be embedded as the top-left submatrix of a $(2m)\times(2m)$ unitary matrix $U$ (see~\cite[Lemma~29]{Aaronson2013} for an explicit construction). We will use this fact multiple times throughout the paper to extend identities proven for unitary matrices to the case of generic matrices.

\subsection{Boson Sampling}

Boson Sampling is a sub-universal model of quantum computation introduced by Aaronson and Arkhipov (AA) in~\cite{Aaronson2013}, which takes as input a Fock state $\ket{\bm1\oplus\bm0}=\ket1^{\otimes n}\otimes\ket0^{\otimes m-n}$, evolves it according to a passive linear operation $\hat U$ over $m$ modes with unitary matrix $U$, and measures the photon number of each output mode. 

With Eq.~(\ref{eq:quantum3}), the probability of detecting $\bm p=(p_1,\dots,p_m)\in\mathbb N^m$ output photons is given by:
\begin{equation}
    \begin{aligned}
        P_\text{BS}(\bm p|n):=&|\langle p_1\dots p_m|\hat U(\ket1^{\otimes n}\otimes\ket 0^{\otimes m-n})|^2\\
        =&\frac{\left|\mathrm{Per}(U_{\bm p,\bm 1\oplus\bm0})\right|^2}{\bm p!}.
    \end{aligned}
\end{equation}
This model of quantum computation, while not believed to be universal, is already capable of outperforming its classical counterparts: AA showed that the output probability distribution $P_\text{BS}$ is hard to sample exactly classically for $m\ge2n$, or the polynomial hierarchy of complexity classes collapses to its third level~\cite{Aaronson2013}. 

Moreover, under additional plausible conjectures, AA showed that this collapse holds for $m=\Theta(n^5\log^2n)$ even if only an efficient classical algorithm for \textit{approximate sampling} exists (i.e.\ a classical algorithm which samples efficiently from a probability distribution that has a small total variation distance with the ideal Boson Sampling output probability distribution $P_\text{BS}$). They further conjectured that this approximate hardness should hold for $m=\Theta(n^2)$. 

The approximate hardness proof for Boson Sampling is based on a matrix-hiding argument which requires a specific regime of parameters $m$ and $n$: given $\delta>0$, the relation between $m$ and $n$ should be such that the distribution $\mathcal H^{n,m}$ of $n\times n$ submatrices of $m\times m$ Haar-random unitary matrices multiplied by $\sqrt m$ is $O(\delta)$-close in total variation distance to the distribution $\mathcal G^{n\times n}$ of $n\times n$ matrices of i.i.d.\ Gaussians (see Theorem 35 in~\cite{Aaronson2013}). In particular, AA showed that, for any $\delta>0$, $\|\mathcal H^{n,m}-\mathcal G^{n\times n}\|_\text{TV}=O(\delta)$ when $m\ge\frac{n^5}\delta\log^2\frac n\delta$.

This result was later refined in~\cite{leverrier2018p} with the bound $\|\mathcal H^{n,m}-\mathcal G^{n\times n}\|_\text{TV}\le\frac{n^3}{2(m-n)}$, which gives a total variation distance $O(\delta)$ for $m\ge\frac{n^3}\delta$.

Finally, in the case of orthogonal matrices rather than unitary matrices, it was shown in~\cite{jiang2017distances} that (see the equation following Eq.~(2.31) in~\cite{jiang2017distances}, with $n\rightarrow m$ and $p=q\rightarrow n$):
\begin{equation}
    \begin{aligned}
        D_\text{KL}(\mathcal H^{n,m},\mathcal G^{n\times n})&\le\frac{n^3}{4(m-n)}-\frac{(m-n)-(n+1)}2\left[\frac{n^3+O(n^2)}{2(m-n)^2}+\frac{n^4+O(n^3)}{3(m-n)^3}\right]\\
        &\le O\!\left(\frac{n^2}m+\frac{n^4}{m^2}+\frac{n^5}{m^3}\right),
    \end{aligned}
\end{equation}
where $D_\text{KL}$ is the Kullback--Leibler distance and where we obtain the second line under the assumption $n=o(m)$. In particular, for $\delta>0$, choosing $m\ge\frac{n^2}{\delta^2}$ implies $D_\text{KL}(\mathcal H^{n,m},\mathcal G^{n\times n})\le O(\delta^2)$, and thus $\|\mathcal H^{n,m}-\mathcal G^{n\times n}\|_\text{TV}\le O(\delta)$ by Pinsker's inequality (see Eq.~(1.3) in~\cite{jiang2017distances}). It was also argued in~\cite{jiang2017distances} that a similar result should hold for unitary matrices.

Summarising these results, approximate hardness of Boson Sampling is proven for passive linear operations described by orthogonal matrices in the regime $m=\Theta(n^2)$~\cite{jiang2017distances} and for unitary matrices in the regime $m=\Theta(n^3)$~\cite{leverrier2018p}.

\subsection{The MacMahon master theorem}
\label{sec:mmmt}

The MacMahon master theorem is an important result in combinatorics which relates the permanent to the determinant:
\begin{theorem*}[MacMahon master theorem~\cite{macmahon2001combinatory}]\label{th:mmmt}
Let $\bm z=(z_1,\dots,z_m)$ be formal variables. For any $m\times m$ matrix $A$,
\begin{equation}
    \sum_{\bm p\in\mathbb N^m}\frac{\bm z^{\bm p}}{\bm p!}\mathrm{Per}(A_{\bm p,\bm p})=\frac1{\mathrm{Det}(I-ZA)},
\end{equation}
where $Z=\mathrm{Diag}(\bm z)$.
\end{theorem*}

\noindent This theorem is particularly useful to derive short proofs of combinatorial identities: it expresses the permanent of an $m\times m$ matrix $A$ with rows and columns repeated in the same way as the coefficient
\begin{equation}
    \mathrm{Per}(A_{\bm p,\bm p})=\bm p![\bm z^{\bm p}]\left(\frac1{\mathrm{Det}(I-ZA)}\right),
\end{equation}
where $Z=\mathrm{Diag}(\bm z)$ and $\bm p\in\mathbb N^m$, while the same permanent may also be expressed as the coefficient
\begin{equation}
    \mathrm{Per}(A_{\bm p,\bm p})=\bm p![\bm z^{\bm p}](A\bm z)^{\bm p}.
\end{equation}
Picking a specific matrix $A$ and a pattern $\bm p$ and computing the above expressions yields combinatorial identities, a famous example being the short proof of Dixon's identity~\cite{dixon1891sum,good1962proofs}, $\sum_{k=0}^{2n}(-1)^k\binom{2n}k^3=(-1)^n\binom{3n}{n,n,n}$, by taking
\begin{equation}
    A=\begin{pmatrix}0&1&-1\\-1&0&1\\1&-1&0\end{pmatrix},
\end{equation}
and $\bm p=(2n,2n,2n)$ for $n\in\mathbb N^*$.

\medskip

Various generalizations of the MacMahon master theorem have been introduced over the years~\cite{good1962short,garoufalidis2006quantum,konvalinka2007non,tuite2013some,kocharovsky2022hafnian}. 
In physics, this theorem plays an important role in the quantum theory of angular momentum~\cite{chen2001angular} and is also oftentimes interpreted as an instance of the boson-fermion correspondence~\cite{garoufalidis2006quantum}.

\section{Main results}
\label{sec:results}

In this section, we summarize our main findings, which we prove in section~\ref{sec:new}. We obtain the following generalization of the MacMahon master theorem:

\begin{theorem}\label{th:mmmt+}
Let $\bm x=(x_1,\dots,x_m)$ and $\bm y=(y_1,\dots,y_m)$ be formal variables. For all $m\times m$ matrices $A$ and $B$,
\begin{equation}
    \sum_{\bm p,\bm q\in\mathbb N^m}\frac{\bm x^{\bm p}\bm y^{\bm q}}{\bm p!\bm q!}\mathrm{Per}(A_{\bm p,\bm q})\mathrm{Per}(B_{\bm q,\bm p})=\frac1{\mathrm{Det}(I-XAYB)},
\end{equation}
where $X=\mathrm{Diag}(\bm x)$ and $Y=\mathrm{Diag}(\bm y)$.
\end{theorem}

\noindent We further show that this generalization extends to $N$ matrices:

\begin{theorem}\label{th:mmmt++}
Let $N\ge2$. For all $k\in\{1,\dots,N\}$, let $\bm z_k=(z_{k1},\dots,z_{km})$ be formal variables. For all $m\times m$ matrices $A^{(1)},\dots,A^{(N)}$,
\begin{equation}
    \begin{aligned}
        &\sum_{\bm p_1,\dots,\bm p_N\in\mathbb N^m}\prod_{k=1}^N\frac{\bm z_k^{\bm p_k}}{\bm p_k!}\mathrm{Per}(A^{(1)}_{\bm p_1,\bm p_2})\mathrm{Per}(A^{(2)}_{\bm p_2,\bm p_3})\dots\mathrm{Per}(A^{(N)}_{\bm p_N,\bm p_1})\\
        &\quad\quad\quad\quad\quad\quad=\frac1{\mathrm{Det}(I-Z_1A^{(1)}\dots Z_NA^{(N)})},
    \end{aligned}
\end{equation}
where $Z_k=\mathrm{Diag}(\bm z_k)$ for all $k\in\{1,\dots,N\}$.
\end{theorem}

\noindent As a corollary of Theorem~\ref{th:mmmt+}, we obtain:

\begin{corollary}\label{coro:mmmt+}
For any $m\times m$ matrix $A$ and all $\bm p,\bm q\in\mathbb N^m$ with $|\bm p|=|\bm q|=n\in\mathbb N$,
\begin{equation}
        \mathrm{Per}(A_{\bm p,\bm q})=\frac{\bm p!\bm q!}{n!}[\bm x^{\bm p}\bm y^{\bm q}]\left(\bm x^TA\bm y\right)^n.
\end{equation}
\end{corollary}

\noindent As a consequence of Corollary~\ref{coro:mmmt+}, we obtain the following family of generating functions for the permanent:

\begin{theorem}\label{th:gen}
Let $f(z)=\sum_{n=0}^{+\infty}f_nz^n$ be a series. Let $\bm x=(x_1,\dots,x_m)$ and $\bm y=(y_1,\dots,y_m)$ be formal variables. For any $m\times m$ matrix $A$,
\begin{equation}
    f(\bm x^TA\bm y)=\sum_{\substack{n\in\mathbb N,\bm p,\bm q\in\mathbb N^m\\|\bm p|=|\bm q|=n}}f_nn!\frac{\bm x^{\bm p}\bm y^{\bm q}}{\bm p!\bm q!}\mathrm{Per}(A_{\bm p,\bm q}),
\end{equation}
Equivalently,
\begin{equation}
    \partial^{\bm p}_{\bm x}\partial^{\bm q}_{\bm y}f(\bm x^TA\bm y)\big\vert_{\bm x=\bm y=\bm0}=\mathrm{Per}(A_{\bm p,\bm q})\,\partial^n_zf(z)\big\vert_{z=0},
\end{equation}
for all $n\in\mathbb N$ and all $\bm p,\bm q\in\mathbb N^m$ such that $|\bm p|=|\bm q|=n$. As a result, when $f_n\neq0$,
\begin{equation}
    \mathrm{Per}(A_{\bm p,\bm q})=\underset{\bm x,\bm y\in\mathbb T^m}{\mathbb E}\left[\frac{\bm p!\bm q!}{\bm x^{\bm p}\bm y^{\bm q}}\frac{f(\bm x^TA\bm y)}{\partial^n_zf(z)\big\vert_{z=0}}\right],
\end{equation}
where the average is over random vectors with complex coefficients of modulus $1$.
\end{theorem}

\noindent From this theorem we derive various permanent identities, the most notable one being a formula for the sum of two permanents:

\begin{theorem}\label{th:sumper}
For all $m\times m$ matrices $A$ and $B$, all $n\in\mathbb N$ and all $\bm p,\bm q\in\mathbb N^m$ such that $|\bm p|=|\bm q|=n$,
\begin{equation}
    \begin{aligned}
        &\mathrm{Per}(A_{\bm p,\bm q})+\mathrm{Per}(B_{\bm p,\bm q})\\
        &\quad=\sum_{k=0}^{\lfloor\frac n2\rfloor}\frac{(-1)^k}{\binom{n-1}k}\!\!\!\!\!\!\!\sum_{\substack{\bm a+\bm b+\bm c=\bm p\\\bm a'+\bm b'+\bm c'=\bm q\\|\bm a|=|\bm b|=|\bm a'|=|\bm b'|=k}}\!\!\!\!\!\!\!\frac{\bm p!\bm q!}{\bm a!\bm b!\bm c!\bm a'!\bm b'!\bm c'!}\mathrm{Per}(A_{\bm a,\bm a'})\mathrm{Per}(B_{\bm b,\bm b'})\mathrm{Per}((A+B)_{\bm c,\bm c'}).
    \end{aligned}
\end{equation}
In particular, when $\bm p=\bm q=\bm 1$,
\begin{equation}
        \mathrm{Per}(A)+\mathrm{Per}(B)=\sum_{k=0}^{\lfloor\frac n2\rfloor}\frac{(-1)^k}{\binom{n-1}k}\!\!\!\!\!\!\sum_{\substack{\bm a+\bm b+\bm c=\bm 1\\\bm a'+\bm b'+\bm c'=\bm 1\\|\bm a|=|\bm b|=|\bm a'|=|\bm b'|=k}}\!\!\!\!\!\!\!\!\mathrm{Per}(A_{\bm a,\bm a'})\mathrm{Per}(B_{\bm b,\bm b'})\mathrm{Per}((A+B)_{\bm c,\bm c'}).
\end{equation}
\end{theorem}

\noindent Section~\ref{sec:qproofs} is primarily devoted to quantum-inspired proofs of existing permanent identities, including a new proof of the MacMahon master theorem. 
Along the way, we obtain the following inner product formula: 

\begin{lemma}\label{lem:BScat}
For all $\alpha\in\mathbb C$, all $\bm p\in\mathbb N^m$, all $n\le m$, and any passive linear operation $\hat U$ over $m$ modes with unitary matrix $U$,
\begin{equation}
    \begin{aligned}
        \langle\bm p|\hat U(\ket{\mathrm{cat}_\alpha}^{\otimes n}\otimes\ket 0^{\otimes m-n})&=\frac{\alpha^n}{\sqrt{\sinh^n(|\alpha|^2)}}\langle\bm p|\hat U(\ket1^{\otimes n}\otimes\ket 0^{\otimes m-n})\\
        &=\frac{\alpha^n}{\sqrt{\sinh^n(|\alpha|^2)}}\frac{\mathrm{Per}(U_{\bm p,\bm1\oplus\bm0})}{\sqrt{\bm p!}}.
    \end{aligned}
\end{equation}
\end{lemma}

\noindent Lemma~\ref{lem:BScat} has the following direct implications for the hardness of Boson Sampling with input cat states~\cite{rohde2015evidence}, which we discuss in section~\ref{sec:BScat}:

\begin{theorem}\label{th:hardBScat}
Let $m=\text{poly }n\ge2n$ and $\alpha\neq0$. Boson Sampling with input $\ket{\mathrm{cat}_\alpha}^{\otimes n}\otimes\ket 0^{\otimes m-n}$ with $\alpha\in\mathbb C$ is hard to sample exactly classically unless the polynomial hierarchy of complexity classes collapses to its third level.

Moreover, assuming $|\alpha|=\mathcal O(n^{-1/4}\log^{1/4}m)$ Boson Sampling with input cat states is as hard to sample approximately as Boson Sampling with input single-photons, i.e.\ Boson Sampling with input cat states is hard to sample approximately classically unless the polynomial hierarchy of complexity classes collapses to its third level, in the same regime as Boson Sampling, modulo the complexity conjectures introduced by AA in~\cite{Aaronson2013}.
\end{theorem}

\noindent This result extends the arguments of~\cite{rohde2015evidence}---which gave formal proof of classical hardness in the exact sampling case---to the approximate sampling case.

\section{Discussion}
\label{sec:outlook}

In the next sections, we introduce quantum-inspired proofs of permanent identities. This approach allows us to give quantum-operational interpretations of seminal results, such as the MacMahon master theorem~\cite{macmahon2001combinatory}. In particular, we show that this theorem can be seen as two facets of the same bosonic Gaussian amplitude.

This quantum-inspired approach also yields a breadth of new permanent identities. We give some examples of combinatorial applications of these identities in section~\ref{sec:new} and we anticipate that they have many more.
Beyond these purely combinatorial applications, it would be interesting to investigate whether our new identities may be used to obtain more efficient classical algorithms for computing or estimating the permanent. In particular, our Theorem~\ref{th:gen} provides new estimators for the permanent which could be of interest, by minimizing the variance of these estimators over the choice of the analytic function $f$.

We use the formalism of linear optics with noninteracting bosons, but our approach can be applied more generally to linear optics with other types of particles.
For instance, we expect our proof techniques to lead to determinant identities in the fermionic case~\cite{terhal2002classical}, immanant identities in the case of partially distinguishable particles~\cite{shchesnovich2015partial,spivak2022generalized}, and additional permanent identities in the case of generalized bosons~\cite{kuo2022boson}.
Moreover, graphical languages are currently being developed for linear optical quantum computations~\cite{clement2022lov,de2022quantum}, which could lead to graphical proofs of remarkable identities in combination with our approach.

Finally, our Theorem~\ref{th:hardBScat} gives solid complexity-theoretic foundations for the hardness of Boson Sampling with input cat states. Generation of such states has progressed tremendously in the recent years, thanks to circuit QED in particular~\cite{peropadre2016proposal,girvin2017schrodinger,gu2017microwave}.
We hope that these foundations will motivate an experimental demonstration of quantum speedup based on Schr\"odinger's cat states.

\section{Quantum-inspired proofs of permanent identities}
\label{sec:qproofs}

In this section, we derive quantum-inspired proofs of existing permanent identities. Along the way, we obtain a generalization of Glynn's formula~\cite{glynn2010permanent} for the permanent of matrices with repeated rows and columns in Eq.~(\ref{eq:glynnbothrepeated}), as well as a generalization of the Glynn--Kan formula~\cite{huh2022fast} for the permanent of  matrices with repeated rows and columns in Eq.~(\ref{eq:glynnkanbothrepeated}).

\subsection{Glynn's formula}
\label{sec:Glynn}

Glynn's formula for the permanent of an $m\times m$ matrix $A=(a_{ij})_{1\le i,j\le m}$ is~\cite{glynn2010permanent}:
\begin{equation}
    \mathrm{Per}(A)=\frac1{2^{m-1}}\sum_{\substack{\bm x\in\{-1,1\}^m\\x_1=1}}x_1\dots x_m\prod_{i=1}^m\left(\sum_{j=1}^ma_{ij}x_j\right).
\end{equation}
By symmetry, it is equivalent to the identity
\begin{equation}
    \mathrm{Per}(A)=\frac1{2^m}\sum_{\bm x\in\{-1,1\}^m}x_1\dots x_m\prod_{i=1}^m\left(\sum_{j=1}^ma_{ij}x_j\right).
\end{equation}

\begin{proof}[Proof of Glynn's formula]
To prove this identity using quantum-mechanical tools, let us introduce the unnormalized cat state $\ket{\tilde{\mathrm{cat}}_\alpha}:=\frac1{2\alpha}(\ket\alpha-\ket{-\alpha})$, for $\alpha\in\mathbb C$. Using the Fock basis expansion of coherent states (\ref{eq:quantum1}), we have $\lim_{\alpha\rightarrow0}\ket{\tilde{\mathrm{cat}}_\alpha}=\ket1$ in trace distance. Hence, using the Fock basis expansion of $\hat U$ (\ref{eq:quantum3}),
\begin{equation}
    \mathrm{Per}(U)=\lim_{\alpha\rightarrow0}\langle1\dots1|\hat U|\tilde{\mathrm{cat}}_\alpha\dots\tilde{\mathrm{cat}}_\alpha\rangle,
\end{equation}
where $\hat U$ is a passive linear operation over $m$ modes with unitary matrix $U=(u_{ij})_{1\le i,j\le m}$.
We compute the right hand side of this equation:
\begin{align}\label{eq:Ucatoverlap}
        \nonumber\langle1\dots1|\hat U|\tilde{\mathrm{cat}}_\alpha\dots\tilde{\mathrm{cat}}_\alpha\rangle&=\frac1{(2\alpha)^m}\sum_{\bm x\in\{-1,1\}^m}x_1\dots x_m\langle1\dots1|\hat U|x_1\alpha\dots x_m\alpha\rangle\\
        &=\frac1{(2\alpha)^m}\sum_{\bm x\in\{-1,1\}^m}x_1\dots x_m\langle1\dots1|(\alpha U\bm x)_1\dots(\alpha U\bm x)_m\rangle\displaybreak\\
        \nonumber&=\frac1{(2\alpha)^m}\sum_{\bm x\in\{-1,1\}^m}x_1\dots x_me^{-\frac12|\alpha|^2\|U\bm x\|^2}\prod_{i=1}^m(\alpha U\bm x)_i\\
        \nonumber&=\frac{e^{-\frac m2|\alpha|^2}}{2^m}\sum_{\bm x\in\{-1,1\}^m}x_1\dots x_m\prod_{i=1}^m\left(\sum_{j=1}^mu_{ij}x_j\right),
\end{align}
where we used the action of $\hat U$ on coherent states (\ref{eq:quantum5}) in the second line, the Fock basis expansion of coherent states (\ref{eq:quantum1}) in the third line, and the fact that $U$ is unitary in the last line. Taking the limit when $\alpha\rightarrow0$ yields Glynn's formula for unitary matrices.
\end{proof}

With the same calculations, we may obtain a more general version of Eq.~(\ref{eq:Ucatoverlap}): for all $\alpha\in\mathbb C$ and $\bm p\in\mathbb N^m$,
\begin{equation}\label{eq:Ucatoverlapgen1}
    \begin{aligned}
        \langle p_1\dots p_m|\hat U|\tilde{\mathrm{cat}}_\alpha\dots\tilde{\mathrm{cat}}_\alpha\rangle&=\frac1{(2\alpha)^m}\sum_{\bm x\in\{-1,1\}^m}x_1\dots x_me^{-\frac12|\alpha|^2\|U\bm x\|^2}\prod_{i=1}^m\frac{\left[(\alpha U\bm x)_i\right]^{p_i}}{\sqrt{p_i!}}\\
        &=\frac{\alpha^{|\bm p|-m}e^{-\frac m2|\alpha|^2}}{2^m\sqrt{\bm p!}}\sum_{\bm x\in\{-1,1\}^m}x_1\dots x_m\prod_{i=1}^m\left(\sum_{j=1}^mu_{ij}x_j\right)^{\!p_i}\!\!\!\!\!.
    \end{aligned}
\end{equation}
Note that when $|\bm p|<m$ all products in the sum have at least one $x_i$ missing. Hence, by symmetry, $\langle p_1\dots p_m|\hat U|\tilde{\mathrm{cat}}_\alpha\dots\tilde{\mathrm{cat}}_\alpha\rangle=0$ when $|\bm p|<m$.
With the Fock basis expansion of $\hat U$ (\ref{eq:quantum3}), taking the limit when $\alpha\rightarrow0$ yields a version of Glynn's formula for unitary matrices with repeated rows:
\begin{equation}
    \mathrm{Per}(U_{\bm p,\bm 1})=\frac{\delta_{|\bm p|,m}}{2^m}\sum_{\bm x\in\{-1,1\}^m}x_1\dots x_m\prod_{i=1}^m\left(\sum_{j=1}^mu_{ij}x_j\right)^{\!p_i}\!\!\!\!\!,
\end{equation}
where $\delta$ is the Kronecker symbol.

So far, all identities are derived for unitary matrices $U$. In order to retrieve the same identity for a generic matrix $A=(a_{ij})_{1\le i,j\le n}$ of size $n$, we can embed $\frac1{\|A\|}A$ (or $A$ directly if $A=0$) as a submatrix of a unitary matrix $U$ of size $m=2n$~\cite[Lemma~29]{Aaronson2013} and compute:
\begin{equation}
    \frac{\mathrm{Per}(A_{\bm q,\bm 1})}{\sqrt{\bm q!}}=\|A\|^{|\bm q|}\lim_{\alpha\rightarrow0}(\bra{q_1\dots q_n}\otimes\bra0^{\otimes n})\hat U(\ket{\tilde{\mathrm{cat}}_\alpha}^{\otimes n}\otimes\ket 0^{\otimes n}).
\end{equation}
To do so, we compute a slightly more general inner product: for $\alpha\in\mathbb C$, $\bm p\in\mathbb N^m$ and $n\le m$,
\begin{equation}\label{eq:Ucatoverlapgen2}
    \begin{aligned}
        &\langle p_1\dots p_m|\hat U(\ket{\tilde{\mathrm{cat}}_\alpha}^{\otimes n}\otimes\ket 0^{\otimes m-n})\\
        &\quad\quad=\frac1{(2\alpha)^n}\sum_{\bm x\in\{-1,1\}^n}x_1\dots x_n\langle p_1\dots p_m|\hat U|x_1\alpha\dots x_n\alpha\,0\dots0\rangle\\
        &\quad\quad=\frac1{(2\alpha)^n}\sum_{\bm x\in\{-1,1\}^n}x_1\dots x_ne^{-\frac12|\alpha|^2\|U(\bm x\oplus\bm 0)\|^2}\prod_{i=1}^m\frac{\left[(\alpha U(\bm x\oplus0))_i\right]^{p_i}}{\sqrt{p_i!}}\\
        &\quad\quad=\frac{\alpha^{|\bm p|-n}e^{-\frac n2|\alpha|^2}}{2^n\sqrt{\bm p!}}\sum_{\bm x\in\{-1,1\}^n}x_1\dots x_n\prod_{i=1}^m\left(\sum_{j=1}^nu_{ij}x_j\right)^{\!p_i}\!\!\!\!\!.
    \end{aligned}
\end{equation}
In particular, setting $m=2n$ and $\bm p=\bm q\oplus\bm 0$ for $\bm q\in\mathbb N^n$,
\begin{equation}
    (\bra{q_1\dots q_n}\otimes\bra0^{\otimes n})\hat U(\ket{\tilde{\mathrm{cat}}_\alpha}^{\otimes n}\otimes\ket 0^{\otimes n})=\frac{\alpha^{|\bm q|-n}e^{-\frac n2|\alpha|^2}}{2^n\sqrt{\bm q!}}\!\!\!\!\sum_{\bm x\in\{-1,1\}^n}\!\!\!\!x_1\dots x_n\prod_{i=1}^n\left(\sum_{j=1}^nu_{ij}x_j\right)^{\!q_i}\!\!\!\!\!.
\end{equation}
Choosing $u_{ij}=\frac1{\|A\|}a_{ij}$ for $1\le i,j\le n$ and letting $\alpha$ go to $0$ proves the claim:
the left hand side converges to $\frac1{\sqrt{\bm q!}}\mathrm{Per}(U_{\bm q\oplus\bm0,\bm1\oplus\bm0})=\frac1{\|A\|^{|\bm q|}\sqrt{\bm q!}}\mathrm{Per}(A_{\bm q,\bm1})$, while the right hand side converges to $\frac{\delta_{|\bm q|,n}}{\|A\|^{|\bm q|}2^n\sqrt{\bm q!}}\sum_{\bm x\in\{-1,1\}^n}x_1\dots x_n\prod_{i=1}^n\left(\sum_{j=1}^na_{ij}x_j\right)^{\!q_i}$. Hence,
\begin{equation}\label{eq:glynnrepeated}
    \mathrm{Per}(A_{\bm q,\bm1})=\frac{\delta_{|\bm q|,n}}{2^n}\sum_{\bm x\in\{-1,1\}^n}x_1\dots x_n\prod_{i=1}^n\left(\sum_{j=1}^na_{ij}x_j\right)^{\!q_i}\!\!\!\!\!.
\end{equation}
A similar generalization of Glynn's formula for matrices with repeated rows (or columns) based on roots of unity has previously appeared in~\cite{aaronson2014generalizing}. 

Moreover, for any $m\times m$ matrix $A=(a_{i,j})_{1\le i,j\le m}$ and all $\bm p,\bm q\in\mathbb N^m$ with $|\bm p|=|\bm q|=n$, we have the formula
\begin{equation}\label{eq:glynnbothrepeated}
    \mathrm{Per}(A_{\bm p,\bm q})=\frac{\bm q!}{n^m}\sum_{\bm x\in\{1,e^{\frac{2i\pi}n},\dots,e^{\frac{2i(n-1)\pi}n}\}^m}\bm x^{-\bm q}(A\bm x)^{\bm p},
\end{equation}
which is a version of Glynn's formula for matrices with repeated rows \textit{and} columns. This formula, which has previously appeared in~\cite{chin2018generalized,yung2019universal,shchesnovich2020classical}, can be easily proven by expanding the product $(A\bm x)^{\bm p}=\prod_{i=1}^m\left(\sum_{j=1}^ma_{ij}x_j\right)^{p_i}$ and using properties of roots of unity. We give a quantum-inspired proof in what follows. 

For all $n\ge1$, all $q\le n$, and all $\alpha\in\mathbb C$, we define the following unnormalised superposition of coherent states:
\begin{equation}\label{eq:sigmaalphaqn}
    \ket{\sigma^{q,n}_\alpha}:=\frac{\sqrt{q!}}{n\alpha^q}\sum_{k=0}^{n-1}e^{-\frac{2ikq\pi}n}\ket{e^{\frac{2ik\pi}n}\alpha}.
\end{equation}
Using the definition of coherent states (\ref{eq:coh}), we obtain
\begin{equation}
    \begin{aligned}
        \ket{\sigma^{q,n}_\alpha}&=\frac{\sqrt{q!}e^{-\frac12|\alpha|^2}}{n\alpha^q}\sum_{p\ge0}\frac{\alpha^p}{\sqrt{p!}}\left(\sum_{k=0}^{n-1}e^{\frac{2ik(p-q)\pi}n}\right)\ket p\\
        &=\sqrt{q!}e^{-\frac12|\alpha|^2}\sum_{l\ge0}\frac{\alpha^{ln}}{\sqrt{(ln+q)!}}\ket{ln+q}.
    \end{aligned}
\end{equation}
where we used $\sum_{k=0}^{n-1}e^{\frac{2ik(p-q)\pi}n}=n$ if $p=q\mod n$ and $0$ otherwise in the second line. Hence, $\lim_{\alpha\rightarrow0}\ket{\sigma^{q,n}_\alpha}=\ket q$ in trace distance. Using the Fock basis expansion of $\hat U$ (\ref{eq:quantum3}), we thus have, for all $\bm p,\bm q\in\mathbb N^m$ with $|\bm p|=|\bm q|=n\ge1$,
\begin{equation}
    \frac{\mathrm{Per}(U_{\bm p,\bm q})}{\sqrt{\bm p!\bm q!}}=\lim_{\alpha\rightarrow0}\langle p_1\dots p_m|\hat U\bigotimes_{j=1}^m\ket{\sigma^{q_j,n}_\alpha},
\end{equation}
where $\hat U$ is a passive linear operation over $m$ modes with unitary matrix $U=(u_{ij})_{1\le i,j\le m}$.
We compute the right hand side of this equation:
\begin{equation}\label{eq:computesigma}
    \begin{aligned}
        \langle p_1\dots p_m|\hat U\bigotimes_{j=1}^m\ket{\sigma^{q_j,n}_\alpha}&=\frac{\sqrt{\bm q!}}{n^m\alpha^n}\sum_{k_1,\dots,k_m=0}^{n-1}\prod_{j=1}^me^{-\frac{2ik_jq_j\pi}n}\langle p_1\dots p_m|\hat U\bigotimes_{j=1}^m\ket{e^{\frac{2ik_j\pi}n}\alpha}\\
        &=\frac{\sqrt{\bm q!}}{n^m\alpha^n}\sum_{\bm x\in\{1,e^{\frac{2i\pi}n},\dots,e^{\frac{2i(n-1)\pi}n}\}^m}\!\!\bm x^{-\bm q}\langle p_1\dots p_m|\hat U\ket{\alpha x_1\dots\alpha x_m}\\
        &=\frac{\sqrt{\bm q!}e^{-\frac m2|\alpha|^2}}{\sqrt{\bm p!}n^m}\sum_{\bm x\in\{1,e^{\frac{2i\pi}n},\dots,e^{\frac{2i(n-1)\pi}n}\}^m}\!\!\bm x^{-\bm q}(U\bm x)^{\bm p}.
    \end{aligned}
\end{equation}
where we used the definition of $\ket{\sigma_\alpha^{q,n}}$ (\ref{eq:sigmaalphaqn}) in the first line and where the derivation for the last line is identical to that of Eqs.~(\ref{eq:Ucatoverlap}-\ref{eq:Ucatoverlapgen1}). Letting $\alpha$ go to $0$ proves Eq.~(\ref{eq:glynnbothrepeated}) for unitary matrices. Finally, the case of a generic nonzero $m\times m$ matrix $A$ is obtained as in Eq.~(\ref{eq:Ucatoverlapgen2}) by embedding $\frac1{\|A\|}A$ as a submatrix of a unitary matrix $U$ of size $2m$ and computing
\begin{equation}
    \left(\langle p_1\dots p_m|\otimes\bra0^{\otimes m}\right)\hat U\left(\bigotimes_{j=1}^m\ket{\sigma^{q_j,n}_\alpha}\otimes\ket0^{\otimes m}\right)\!.
\end{equation}

\subsection{The Glynn--Kan formula}

A symmetrized version of Glynn's formula for the permanent of an $m\times m$ matrix $A$ has been recently derived~\cite{huh2022fast} under the name Glynn--Kan formula:
\begin{equation}\label{eq:glynnkan}
    \mathrm{Per}(A)=\frac1{4^mm!}\sum_{\bm x,\bm y\in\{-1,1\}^m}x_1\dots x_my_1\dots y_m(\bm x^TA\bm y)^m.
\end{equation}
While the Glynn--Kan formula gives a slower classical algorithm than the Glynn formula for computing the permanent, it is motivated by a connection with quantum algorithms for estimating the permanent~\cite{huh2022fast}.

\begin{proof}[Proof of the Glynn--Kan formula]
To prove this identity using quantum-mechanical tools, we again use the unnormalized cat state $\ket{\tilde{\mathrm{cat}}_\alpha}:=\frac1{2\alpha}(\ket\alpha-\ket{-\alpha})$, as in the previous section.  Using the Fock basis expansion of coherent states (\ref{eq:quantum1}), we have $\lim_{\alpha\rightarrow0}\ket{\tilde{\mathrm{cat}}_\alpha}=\ket1$ in trace distance. Hence, using the Fock basis expansion of $\hat U$ (\ref{eq:quantum3}),
\begin{equation}
    \mathrm{Per}(U)=\lim_{\alpha\rightarrow0}\langle\tilde{\mathrm{cat}}_\alpha\dots\tilde{\mathrm{cat}}_\alpha|\hat U|\tilde{\mathrm{cat}}_\alpha\dots\tilde{\mathrm{cat}}_\alpha\rangle,
\end{equation}
where $\hat U$ is a passive linear operation over $m$ modes with unitary matrix $U$.
We compute the right hand side of this equation:
\begin{align}\label{eq:Ucatoverlap2}
            \nonumber&\langle\tilde{\mathrm{cat}}_\alpha\dots\tilde{\mathrm{cat}}_\alpha|\hat U|\tilde{\mathrm{cat}}_\alpha\dots\tilde{\mathrm{cat}}_\alpha\rangle\\
            \nonumber&\quad\quad=\frac1{4^m|\alpha|^{2m}}\sum_{\bm x,\bm y\in\{-1,1\}^m}x_1\dots x_my_1\dots y_m\langle x_1\alpha\dots x_m\alpha|\hat U|y_1\alpha\dots y_m\alpha\rangle\\
            \nonumber&\quad\quad=\frac1{4^m|\alpha|^{2m}}\sum_{\bm x,\bm y\in\{-1,1\}^m}x_1\dots x_my_1\dots y_m\langle x_1\alpha\dots x_m\alpha|(\alpha U\bm y)_1\dots(\alpha U\bm y)_m\rangle\\
            &\quad\quad=\frac1{4^m|\alpha|^{2m}}\sum_{\bm x,\bm y\in\{-1,1\}^m}x_1\dots x_my_1\dots y_me^{-\frac12|\alpha|^2\|\bm x\|^2}e^{-\frac12|\alpha|^2\|U\bm y\|^2}\prod_{i=1}^me^{x_i\alpha^*(\alpha U\bm y)_i}\\
            \nonumber&\quad\quad=\frac{e^{-m|\alpha|^2}}{4^m|\alpha|^{2m}}\sum_{\bm x,\bm y\in\{-1,1\}^m}x_1\dots x_my_1\dots y_me^{|\alpha|^2\bm x^TU\bm y}\\
            \nonumber&\quad\quad=\frac{e^{-m|\alpha|^2}}{4^m}\sum_{\bm x,\bm y\in\{-1,1\}^m}x_1\dots x_my_1\dots y_m\sum_{k=0}^{+\infty}\frac{|\alpha|^{2k-2m}}{k!}(\bm x^TU\bm y)^k,
\end{align}
where we used the action of $\hat U$ on coherent states (\ref{eq:quantum5}) in the third line, the overlap between coherent states (\ref{eq:quantum1+}) in the fourth line, and the fact that $U$ is unitary in the fifth line. For $k<m$, all products in the expansion of $(\bm x^TU\bm y)^k$ have at least one $x_i$ missing. Hence, by symmetry, the terms for $k<m$ vanish and we obtain
\begin{equation}\label{eq:Ucatoverlap2gen}
    \langle\tilde{\mathrm{cat}}_\alpha\dots\tilde{\mathrm{cat}}_\alpha|\hat U|\tilde{\mathrm{cat}}_\alpha\dots\tilde{\mathrm{cat}}_\alpha\rangle=\frac{e^{-m|\alpha|^2}}{4^m}\!\!\!\!\!\!\sum_{\bm x,\bm y\in\{-1,1\}^m}\!\!\!\!x_1\dots x_my_1\dots y_m\sum_{l=0}^\infty \frac{|\alpha|^{2l}}{(m+l)!}(\bm x^TU\bm y)^{m+l}.
\end{equation}
Finally, taking the limit when $\alpha$ goes to $0$ yields
\begin{equation}\label{eq:GKU}
    \mathrm{Per}(U)=\frac1{4^mm!}\sum_{\bm x,\bm y\in\{-1,1\}^m}x_1\dots x_my_1\dots y_m(\bm x^TU\bm y)^m.
\end{equation}
This proves the Glynn--Kan formula for a unitary matrix $U$. Once again, the proof extends straightforwardly to any matrix $A$ of size $n$ by embedding $\frac1{\|A\|}A$ as a submatrix of a unitary matrix $U$ of size $2n$ and computing $\lim_{\alpha\rightarrow0}(\bra{\tilde{\mathrm{cat}}_\alpha}^{\otimes n}\otimes\bra0^{\otimes n})\hat U(\ket{\tilde{\mathrm{cat}}_\alpha}^{\otimes n}\otimes\ket0^{\otimes n})$.
\end{proof}

Similar to Eq.~(\ref{eq:glynnbothrepeated}), we obtain a version of the Glynn--Kan formula for matrices with repeated rows and columns using roots of unity as
\begin{equation}\label{eq:glynnkanbothrepeated}
    \mathrm{Per}(A_{\bm p,\bm q})=\frac{\bm p!\bm q!}{n^{2m}n!}\sum_{\bm x,\bm y\in\{1,e^{\frac{2i\pi}n},\dots,e^{\frac{2i(n-1)\pi}n}\}^m}\bm x^{-\bm p}\bm y^{-\bm q}(\bm x^TA\bm y)^n,
\end{equation}
for any $m\times m$ matrix $A=(a_{ij})_{1\le i,j\le m}$ and all $\bm p,\bm q\in\mathbb N^m$ with $|\bm p|=|\bm q|=n$. This formula, which appears to be new, can be easily proven using our Corollary~\ref{coro:mmmt+} by expanding the product $(\bm x^TA\bm y)^n$ and using properties of roots of unity. We give a quantum-inspired proof in what follows. 

For all $n\ge1$, all $q\le n$, and all $\alpha\in\mathbb C$, we again use the unnormalised superposition of coherent states as in Eq.~(\ref{eq:sigmaalphaqn}): $\ket{\sigma^{q,n}_\alpha}:=\frac{\sqrt{q!}}{n\alpha^q}\sum_{k=0}^{n-1}e^{-\frac{2ikq\pi}n}\ket{e^{\frac{2ik\pi}n}\alpha}$. Recall that $\lim_{\alpha\rightarrow0}\ket{\sigma^{q,n}_\alpha}=\ket q$ in trace distance. Using the Fock basis expansion of $\hat U$ (\ref{eq:quantum3}), we thus have, for all $\bm p,\bm q\in\mathbb N^m$ with $|\bm p|=|\bm q|=n\ge1$,
\begin{equation}
    \frac{\mathrm{Per}(U_{\bm p,\bm q})}{\sqrt{\bm p!\bm q!}}=\lim_{\alpha\rightarrow0}\bigotimes_{j=1}^m\bra{\sigma^{p_j,n}_\alpha}\hat U\bigotimes_{j=1}^m\ket{\sigma^{q_j,n}_\alpha},
\end{equation}
where $\hat U$ is a passive linear operation over $m$ modes with unitary matrix $U=(u_{ij})_{1\le i,j\le m}$.
We compute the right hand side of this equation:
\begin{equation}
    \begin{aligned}
        \bigotimes_{j=1}^m\bra{\sigma^{p_j,n}_\alpha}\hat U\bigotimes_{j=1}^m\ket{\sigma^{q_j,n}_\alpha}&=\frac{\sqrt{\bm p!\bm q!}}{n^{2m}|\alpha|^{2n}}\sum_{\substack{k_1,\dots,k_m=0\\l_1,\dots l_m=0}}^{n-1}\prod_{j=1}^me^{-\frac{2ik_jp_j\pi}n}e^{-\frac{2il_jq_j\pi}n}\bigotimes_{j=1}^m\bra{e^{\frac{2ik_j\pi}n}\alpha}\hat U\bigotimes_{j=1}^m\ket{e^{\frac{2il_j\pi}n}\alpha}\\
        &=\frac{\sqrt{\bm p!\bm q!}}{n^{2m}|\alpha|^{2n}}\sum_{\bm x,\bm y\in\{1,e^{\frac{2i\pi}n},\dots,e^{\frac{2i(n-1)\pi}n}\}^m}\!\!\bm x^{-\bm p}\bm y^{-\bm q}\langle x_1\alpha\dots x_m\alpha|\hat U|y_1\alpha\dots y_m\alpha\rangle\\
        &=\frac{\sqrt{\bm p!\bm q!}e^{-m|\alpha|^2}}{n^{2m}}\sum_{\bm x\in\{1,e^{\frac{2i\pi}n},\dots,e^{\frac{2i(n-1)\pi}n}\}}\!\!\bm x^{-\bm p}\bm y^{-\bm q}\sum_{l=0}^\infty \frac{|\alpha|^{2l}}{(n+l)!}(\bm x^TU\bm y)^{n+l}.
    \end{aligned}
\end{equation}
where we used the definition of $\ket{\sigma_\alpha^{q,n}}$ (\ref{eq:sigmaalphaqn}) in the first line and where the derivation for the last line is identical to that of Eqs.~(\ref{eq:Ucatoverlap2}-\ref{eq:Ucatoverlap2gen}). Letting $\alpha$ go to $0$ proves Eq.~(\ref{eq:glynnkanbothrepeated}) for unitary matrices. 
Once again, the proof extends straightforwardly to any nonzero matrix $A$ of size $m$ by embedding $\frac1{\|A\|}A$ as a submatrix of a unitary matrix $U$ of size $2m$ and computing
\begin{equation}
    \left(\bigotimes_{j=1}^m\bra{\sigma^{p_j,n}_\alpha}\otimes\bra0^{\otimes m}\right)\hat U\left(\bigotimes_{j=1}^m\ket{\sigma^{q_j,n}_\alpha}\otimes\ket0^{\otimes m}\right)\!.
\end{equation}

\subsection{The Cauchy--Binet theorem}

The Cauchy--Binet theorem for the permanent expresses the permanent of the product of two $m\times m$ matrices $A$ and $B$ as a sum of products involving permanents of submatrices of these two matrices~\cite{minc1984permanents,percus2012combinatorial}: for all $\bm p,\bm q\in\mathbb N^m$,
\begin{equation}\label{eq:percomposition}
    \mathrm{Per}((AB)_{\bm p,\bm q})=\sum_{\bm k\in\mathbb N^m}\frac1{\bm k!}\mathrm{Per}(A_{\bm p,\bm k})\mathrm{Per}(B_{\bm k,\bm q}).
\end{equation}

\begin{proof}[Proof of the Cauchy--Binet theorem]
Using the Fock basis expansion of passive linear operations (\ref{eq:quantum3}) gives a quick quantum-inspired proof of this identity: for $U$ and $V$ two $m\times m$ unitary matrices and for all $\bm p,\bm q\in\mathbb N^m$,
\begin{align}
        \nonumber\mathrm{Per}((UV)_{\bm p,\bm q})&=\sqrt{\bm p!\bm q!}\bra{\bm p}\widehat{UV}\ket{\bm q}\\
        \nonumber&=\sqrt{\bm p!\bm q!}\bra{\bm p}\hat U\hat V\ket{\bm q}\\
        &=\sqrt{\bm p!\bm q!}\bra{\bm p}\hat U\left(\sum_{\bm k\in\mathbb N^m}\ket{\bm k}\!\bra{\bm k}\right)\hat V\ket{\bm q}\\
        \nonumber&=\sum_{\bm k\in\mathbb N^m}\sqrt{\bm p!\bm q!}\bra{\bm p}\hat U\ket{\bm k}\!\bra{\bm k}\hat V\ket{\bm q}\\
        \nonumber&=\sum_{\bm k\in\mathbb N^m}\frac1{\bm k!}\mathrm{Per}(U_{\bm p,\bm k})\mathrm{Per}(V_{\bm k,\bm q}),
\end{align}
where we used the fact that Fock states form a basis in the third line and where we used Eq.~(\ref{eq:quantum3}) once in the first line and twice in the last line. 

In order to retrieve the formula for generic matrices $A$ and $B$, we can embed $\frac1{\|A\|}A$ as a submatrix of a unitary matrix $U$ and $\frac1{\|B\|}B$ as a submatrix of a unitary matrix $V$ and compute $\mathrm{Per}((UV)_{\bm p\oplus\bm 0,\bm q\oplus\bm 0})$.
\end{proof}

\subsection{Generating functions}
\label{sec:pergen}

The permanent may be seen as a monomial coefficient in the Taylor expansion of various functions~\cite{minc1984permanents,percus2012combinatorial}. For example, the permanent of an $m\times m$ matrix $A=(a_{ij})_{1\le i,j\le m}$ with rows repeated according to $\bm p\in\mathbb N^m$ and columns repeated according to $\bm q\in\mathbb N^m$ is given by the coefficient of $\bm z^{\bm q}=z_1^{q_1}\dots z_m^{q_m}$ in
\begin{equation}
    \bm q!(A\bm z)^{\bm p}=\prod_{i=1}^mq_i!\left(\sum_{j=1}^ma_{ij}z_j\right)^{\!p_i}\!\!\!\!\!.
\end{equation}
This may be thought of as the `monomial version' of Glynn's formula in Eq.~(\ref{eq:glynnrepeated}). Formally:
\begin{equation}\label{eq:permonom}
        \sum_{\bm q\in\mathbb N^m}\frac{\bm z^{\bm q}}{\bm q!}\mathrm{Per}(A_{\bm p,\bm q})=(A\bm z)^{\bm p},
\end{equation}
for formal variables $\bm z=(z_1,\dots,z_m)$.

\begin{proof}[Proof of `Glynn's monomial formula']
To prove this relation with quantum mechanical tools, we fix $\bm p\in\mathbb N^m$, $\bm\alpha\in\mathbb C^m$, a unitary matrix $U$ of size $m$, and we compute:
\begin{align}
        \nonumber\sum_{\bm q\in\mathbb N^m}\frac{\bm\alpha^{\bm q}}{\bm q!}\mathrm{Per}(U_{\bm p,\bm q})&=\sqrt{\bm p!}e^{\frac12\|\bm\alpha\|^2}\bra{\bm p}\hat U\left(\sum_{\bm q\in\mathbb N^m}\ket{\bm q}\!\bra{\bm q}\right)\ket{\bm\alpha}\displaybreak\\
        \nonumber&=\sqrt{\bm p!}e^{\frac12\|\bm\alpha\|^2}\bra{\bm p}\hat U\ket{\bm\alpha}\\
        &=\sqrt{\bm p!}e^{\frac12\|\bm\alpha\|^2}\bra{\bm p}(U\bm\alpha)\rangle\\
        \nonumber&=e^{\frac12\|\bm\alpha\|^2-\frac12\|U\bm\alpha\|^2}(U\bm\alpha)^{\bm p}\\
        \nonumber&=(U\bm\alpha)^{\bm p},
\end{align}
where we used the Fock basis expansions of coherent states (\ref{eq:quantum1}) and of $\hat U$ (\ref{eq:quantum3}) in the first line, the fact that Fock states form a basis in the second line, the action of $\hat U$ on coherent states (\ref{eq:quantum5}) in the third line, Eq.~(\ref{eq:quantum1}) again in the fourth line, and the fact that $U$ is unitary in the last line. Once again, the relation for a generic nonzero matrix $A$ is obtained by embedding $\frac1{\|A\|}A$ as a submatrix of a unitary matrix $U$.
\end{proof}

Another generating function for the permanent is due to Jackson~\cite{jackson1977unification}: for formal variables $\bm x=(x_1,\dots,x_m),\bm y=(y_1,\dots,y_m)$,
\begin{equation}
    \sum_{\bm p,\bm q\in\mathbb N^m}\frac{\bm x^{\bm p}\bm y^{\bm q}}{\bm p!\bm q!}\mathrm{Per}(A_{\bm p,\bm q})=e^{\bm x^TA\bm y}.
\end{equation}
Note that this relation implies Corollary~\ref{coro:mmmt+}, by expanding the Taylor series of the exponential and considering the $\bm x^{\bm p}\bm y^{\bm q}$ coefficient when $|\bm p|=|\bm q|=n\in\mathbb N^*$ (an alternative proof of this result is given in the next section):
\begin{equation}
    \mathrm{Per}(A_{\bm p,\bm q})=\frac{\bm p!\bm q!}{n!}[\bm x^{\bm p}\bm y^{\bm q}](\bm x^TA\bm y)^n.
\end{equation}

\begin{proof}[Proof of Jackson's formula]
Jackson's formula may be derived using quantum mechanical tools in a similar way: for $\bm\alpha,\bm\beta\in\mathbb C^m$ and a unitary matrix $U$,
\begin{equation}
    \begin{aligned}
        \sum_{\bm p,\bm q\in\mathbb N^m}\frac{\bm\alpha^{\bm p}\bm\beta^{\bm q}}{\bm p!\bm q!}\mathrm{Per}(U_{\bm p,\bm q})&=e^{\frac12\|\bm\alpha\|^2+\frac12\|\bm\beta\|^2}\bra{\bm\alpha^*}\left(\sum_{\bm p\in\mathbb N^m}\ket{\bm p}\!\bra{\bm p}\right)\hat U\left(\sum_{\bm q\in\mathbb N^m}\ket{\bm q}\!\bra{\bm q}\right)\ket{\bm\beta}\\
        &=e^{\frac12\|\bm\alpha\|^2+\frac12\|\bm\beta\|^2}\bra{\bm\alpha^*}\hat U\ket{\bm\beta}\\
        &=e^{\frac12\|\bm\alpha\|^2+\frac12\|\bm\beta\|^2}\bra{\bm\alpha^*}U\bm\beta\rangle\\
        &=e^{\frac12\|\bm\beta\|^2-\frac12\|U\bm\beta\|^2}e^{\bm\alpha^TU\bm\beta}\\
        &=e^{\bm\alpha^TU\bm\beta},
    \end{aligned}
\end{equation}
where we used the Fock basis expansions of coherent states (\ref{eq:quantum1}) and of $\hat U$ (\ref{eq:quantum3}) in the first line, the fact that Fock states form a basis in the second line, the action of $\hat U$ on coherent states (\ref{eq:quantum5}) in the third line, the overlap between coherent states (\ref{eq:quantum1+}) in the fourth line, and the fact that $U$ is unitary in the last line. Once again, the relation for a generic nonzero matrix $A$ is obtained by embedding $\frac1{\|A\|}A$ as a submatrix of a unitary matrix $U$.
\end{proof}

We conclude this section with arguably one of the most remarkable permanent identities, the MacMahon master theorem~\cite{macmahon2001combinatory}, which relates the permanent and the determinant through a generating function (see section~\ref{sec:mmmt}):
\begin{equation}\label{eq:mmmt}
    \sum_{\bm p\in\mathbb N^m}\frac{\bm z^{\bm p}}{\bm p!}\mathrm{Per}(A_{\bm p,\bm p})=\frac1{\mathrm{Det}(I-ZA)},
\end{equation}
where $Z=\mathrm{Diag}(\bm z)$, with $\bm z=(z_1,\dots,z_m)$ formal variables.

\begin{proof}[Proof of the MacMahon master theorem]
To prove this relation with quantum mechanical tools, we show that the MacMahon master theorem describes two different ways of computing an inner product between two Gaussian states.

For $\lambda\in\mathbb C$ and $|\lambda|<1$, we make use of (unnormalized) two-mode squeezed states of the form $\ket{\lambda}=\sum_{p\ge0}\lambda^p\ket{pp}$. We write $\ket{\bm\lambda}$ with $\bm\lambda=(\lambda_1,\dots,\lambda_m)\in\mathbb C^m$ a tensor product of $m$ two-mode squeezed states (note that we are associating mode $1$ with mode $m+1$, mode $2$ with mode $m+2$, and so on).
For $U$ a unitary matrix and $\bm\lambda,\bm\mu\in\mathbb C^m$, with $|\lambda_k|<1$ and $|\mu_k|<1$ for all $k\in\{1,\dots,m\}$, we compute, using the Fock basis expansion of $\hat U$ (\ref{eq:quantum3}):
\begin{equation}\label{eq:mmmt1}
        \sum_{\bm p\in\mathbb N^m}\frac{\bm\lambda^{\bm p}\bm\mu^{\bm p}}{\bm p!}\mathrm{Per}(U_{\bm p,\bm p})=\sum_{\bm p\in\mathbb N^m}\bm\lambda^{\bm p}\bm\mu^{\bm p}\bra{\bm p}\hat U\ket{\bm p}.
\end{equation}
Using the Fock basis expansion of two-mode squeezed states (\ref{eq:quantum2}), we have
\begin{equation}
    \bm\lambda^{\bm p}=\bra{\bm\lambda^*}\bm p\bm p\rangle.
\end{equation}
Moreover, a quick computation in Fock basis shows that
\begin{equation}
    \bm\mu^{\bm p}\bra{\bm p}\hat U\ket{\bm p}=\bra{\bm p\bm p}\hat I\otimes\hat U\ket{\bm\mu}.
\end{equation}
Hence, Eq.~(\ref{eq:mmmt1}) rewrites
\begin{equation}\label{eq:mmmt2}
    \begin{aligned}
        \sum_{\bm p\in\mathbb N^m}\frac{\bm\lambda^{\bm p}\bm\mu^{\bm p}}{\bm p!}\mathrm{Per}(U_{\bm p,\bm p})&=\sum_{\bm p\in\mathbb N^m}\bra{\bm\lambda^*}\bm p\bm p\rangle\!\bra{\bm p\bm p}\hat I\otimes\hat U\ket{\bm\mu}\\
        &=\bra{\bm\lambda^*}\left(\sum_{\bm p\in\mathbb N^m}\ket{\bm p\bm p}\!\bra{\bm p\bm p}\right)\hat I\otimes\hat U\ket{\bm\mu}\\
        &=\bra{\bm\lambda^*}\hat I\otimes\hat U\ket{\bm\mu},
    \end{aligned}
\end{equation}
where in the last line we used $\bra{\bm\lambda^*}\left(\sum_{\bm p\in\mathbb N^m}\ket{\bm p\bm p}\!\bra{\bm p\bm p}\right)=\bra{\bm\lambda^*}$, which can be checked in Fock basis. 

We now compute the Gaussian inner product $\bra{\bm\lambda^*}\hat I\otimes\hat U\ket{\bm\mu}$. From the action of passive linear operations on creation operators (\ref{eq:actioncrea}), the unitary matrix associated to the passive linear operation $\hat I\otimes\hat U$ is $I\oplus U$. With the Gaussian overlap derived in Eqs.~(\ref{eq:quantum6}-\ref{eq:Vw}), we thus obtain:
\begin{equation}\label{eq:mmmt3}
    \bra{\bm\lambda^*}\hat I\otimes\hat U\ket{\bm\mu}=\frac1{\sqrt{\mathrm{Det}(V_{I\oplus U}(\bm\lambda,\bm\mu))}},
\end{equation}
where
\begin{equation}
    V_{I\oplus U}(\bm\lambda,\bm\mu)=\begin{pmatrix}-(I\oplus U)V_{\bm\mu}(I\oplus U)^T&I_{2m}\\I_{2m}&-V_{\bm\lambda}\end{pmatrix},
\end{equation}
where for all $\bm w\in\mathbb C^m$,
\begin{equation}
    V_{\bm w}=\begin{pmatrix}0_m&\mathrm{Diag}(\bm w)\\\mathrm{Diag}(\bm w)&0_m\end{pmatrix}.
\end{equation}
We have
\begin{align}\label{eq:detV}
    \nonumber\!\mathrm{Det}(V_{I\oplus U}(\bm\lambda,\bm\mu))&=\mathrm{Det}(I-(I\oplus U)V_{\bm\mu}(I\oplus U)^TV_{\bm\lambda})\\
    \nonumber&=\mathrm{Det}\!\left(I-(I\oplus U)\!\begin{pmatrix}0_m&\mathrm{Diag}(\bm\mu)\\\mathrm{Diag}(\bm\mu)&0_m\end{pmatrix}\!(I\oplus U)^T\!\!\begin{pmatrix}0_m&\mathrm{Diag}(\bm\lambda)\\\mathrm{Diag}(\bm\lambda)&0_m\end{pmatrix}\!\right)\!\!\!\\
    &=\mathrm{Det}\!\begin{pmatrix}I-\mathrm{Diag}(\bm\mu)U^T\mathrm{Diag}(\bm\lambda)&0_m\\0_m&I-U\mathrm{Diag}(\bm\mu)\mathrm{Diag}(\bm\lambda)\end{pmatrix}\displaybreak\\
    \nonumber&=\mathrm{Det}(I-\mathrm{Diag}(\bm\mu)U^T\mathrm{Diag}(\bm\lambda))\,\mathrm{Det}(I-U\mathrm{Diag}(\bm\mu)\mathrm{Diag}(\bm\lambda))\\
    \nonumber&=\mathrm{Det}(I-\mathrm{Diag}(\bm\lambda)\mathrm{Diag}(\bm\mu)U)^2,
\end{align}
where the last line is obtained by using Sylvester's determinant theorem $\mathrm{Det}(I+MN)=\mathrm{Det}(I+NM)$, together with the transpose for the first determinant and $\mathrm{Diag}(\bm\mu)\mathrm{Diag}(\bm\lambda)=\mathrm{Diag}(\bm\lambda)\mathrm{Diag}(\bm\mu)$ for the second determinant. Combining Eqs.~(\ref{eq:mmmt2}), (\ref{eq:mmmt3}) and~(\ref{eq:detV}), we obtain, $\bm\lambda,\bm\mu\in\mathbb C^m$, with $|\lambda_k|<1$ and $|\mu_k|<1$ for all $k\in\{1,\dots,m\}$,
\begin{equation}
    \sum_{\bm p\in\mathbb N^m}\frac{\bm\lambda^{\bm p}\bm\mu^{\bm p}}{\bm p!}\mathrm{Per}(U_{\bm p,\bm p})=\bra{\bm\lambda^*}\hat I\otimes\hat U\ket{\bm\mu}=\frac1{\mathrm{Det}(I-\mathrm{Diag}(\bm\lambda)\mathrm{Diag}(\bm\mu)U)},
\end{equation}
where we used the value at $\bm\lambda=\bm\mu=0$ and the fact that the left hand side is a continuous function of $\bm\lambda$ and $\bm\mu$ to determine the sign of the square root.
Replacing $(\lambda_1\mu_1,\dots,\lambda_m\mu_m)$ by the formal variables $\bm z=(z_1,\dots,z_m)$ concludes the proof. The relation for a generic nonzero matrix $A$ of size $n$ is obtained by embedding $\frac1{\|A\|}A$ as a submatrix of a unitary matrix $U$ of size $2n$ and taking $\bm z=(z_1,\dots,z_n,0,\dots,0)$.
\end{proof}

\section{New quantum-inspired permanent identities}
\label{sec:new}

In this section, we derive new quantum-inspired identities for the permanent and we give some combinatorial applications of these identities.

\subsection{Generalizations of the MacMahon master theorem}

In this section, we introduce new generalizations of the MacMahon master theorem (see section~\ref{sec:mmmt}). 

We first derive new quantum-inspired identities involving the permanent of even-sized matrices: for $M$ a $(2m)\times(2m)$ matrix,
\begin{equation}\label{eq:new1}
    \mathrm{Per}(M)=[z^m]\left(\frac1{4^m}\sum_{\bm x,\bm y\in\{-1,1\}^m}\frac{x_1\dots x_my_1\dots y_m}{\sqrt{\mathrm{Det}(I-zV_{\bm x}MV_{\bm y}M^T)}}\right),
\end{equation}
and, for formal variables $\bm x=(x_1,\dots,x_m),\bm y=(y_1,\dots,y_m)$,
\begin{equation}\label{eq:new2}
    \sum_{\bm p,\bm q\in\mathbb N^m}\frac{\bm x^{\bm p}\bm y^{\bm q}}{\bm p!\bm q!}\mathrm{Per}(M_{\bm p\oplus\bm p,\bm q\oplus\bm q})=\frac1{\sqrt{\mathrm{Det}(I-V_{\bm x}MV_{\bm y}M^T)}},
\end{equation}
where for all $\bm w=(w_1,\dots,w_m)$,
\begin{equation}
    V_{\bm w}=\begin{pmatrix}0_m&\mathrm{Diag}(\bm w)\\\mathrm{Diag}(\bm w)&0_m\end{pmatrix}.
\end{equation}
As a direct consequence, when $M=A\oplus B$, with $A$ and $B$ two $m\times m$ matrices, we obtain Theorem~\ref{th:mmmt+}:
\begin{equation}\label{eq:new3}
    \sum_{\bm p,\bm q\in\mathbb N^m}\frac{\bm x^{\bm p}\bm y^{\bm q}}{\bm p!\bm q!}\mathrm{Per}(A_{\bm p,\bm q})\mathrm{Per}(B_{\bm p,\bm q})=\frac1{\mathrm{Det}(I-XAYB^T)},
\end{equation}
where $X=\mathrm{Diag}(x_1,\dots,x_m)$ and $Y=\mathrm{Diag}(y_1,\dots,y_m)$. 
In particular, $\mathrm{Per}(A)\mathrm{Per}(B)$ is given by the $x_1\dots x_my_1\dots y_m$ coefficient of $1/\mathrm{Det}(I-XAYB^T)$. As a consequence, we obtain Corollary~\ref{coro:mmmt+}: for all $n\in\mathbb N^*$ and all $\bm p,\bm q\in\mathbb N^m$ with $|\bm p|=|\bm q|=n$,
\begin{equation}\label{eq:new5}
    \mathrm{Per}(A_{\bm p,\bm q})=\frac{\bm p!\bm q!}{n!}[\bm x^{\bm p}\bm y^{\bm q}]\left(\bm x^TA\bm y\right)^n.
\end{equation}
We also generalize Eq.~(\ref{eq:new3}) to the case of $N$ matrices to obtain Theorem~\ref{th:mmmt++}: for all $n\times n$ matrices $A^{(1)},\dots,A^{(N)}$,
\begin{equation}\label{eq:new6}
    \begin{aligned}
        &\sum_{\bm p_1,\dots,\bm p_N\in\mathbb N^m}\prod_{k=1}^N\frac{\bm z_k^{\bm p_k}}{\bm p_k!}\mathrm{Per}(A^{(1)}_{\bm p_1,\bm p_2})\mathrm{Per}(A^{(2)}_{\bm p_2,\bm p_3})\dots\mathrm{Per}(A^{(N)}_{\bm p_N,\bm p_1})\\
        &\quad\quad\quad\quad\quad\quad=\frac1{\mathrm{Det}(I-Z_1A^{(1)}\dots Z_NA^{(N)})},
    \end{aligned}
\end{equation}
where $\bm z_k=(z_{k1},\dots,z_{km})$ are formal variables and $Z_k=\mathrm{Diag}(\bm z_k)$ for all $k\in\{1,\dots,N\}$.

\medskip

\begin{proof}[Proofs of Theorem~\ref{th:mmmt+}, Corollary~\ref{coro:mmmt+}, and Theorem~\ref{th:mmmt++}] In section~\ref{sec:Glynn}, we have obtained a linear optical proof of Glynn's formula for the permanent using the fact that a cat state $\ket{\mathrm{cat}_\alpha}$ of small amplitude $\alpha$ approximates a single-photon Fock state $\ket 1$.
Hereafter, we consider another approximation of Fock states using superpositions of two-mode squeezed states. Such superpositions were recently studied in~\cite{cardoso2021superposition} in the context of quantum sensing. 

To prove the identity in Eq.~(\ref{eq:new1}), we introduce $\ket{\lambda^-}:=\frac1{2\lambda}\left(\ket{\lambda}-\ket{-\lambda}\right)$, for $|\lambda|<1$, where $\ket{\lambda}=\sum_{p\ge0}\lambda^p\ket{pp}$ is an unnormalized two-mode squeezed state.
We have $\lim_{\lambda\rightarrow0}\ket{\lambda^-}=\ket{11}$ in trace distance. Hence, for $U$ a $(2m)\times(2m)$ unitary matrix we have with the Fock basis expansion of $\hat U$ (\ref{eq:quantum3}):
\begin{equation}
    \mathrm{Per}(U)=\lim_{\lambda\rightarrow0}\prescript{\otimes m}{}{\bra{\lambda^-}}\hat U\ket{\lambda^-}^{\otimes m}.
\end{equation}
For all $\lambda\in\mathbb R$, with $|\lambda|<1$, let us compute
\begin{equation}\label{eq:perTMSS}
    \prescript{\otimes m}{}{\bra{\lambda^-}}\hat U\ket{\lambda^-}^{\otimes m}=\frac1{4^m\lambda^{2m}}\sum_{\bm x,\bm y\in\{-1,1\}^m}x_1\dots x_my_1\dots y_mg_U(\lambda\bm x,\lambda\bm y),
\end{equation}
where we have defined for all $\bm\lambda,\bm\mu\in\mathbb C^m$ (note that we associate mode $1$ with mode $m+1$, mode $2$ with mode $m+2$, and so on)
\begin{equation}\label{eq:defgU}
        g_U(\bm\lambda,\bm\mu):=\bra{\bm\lambda^*}\hat U\ket{\bm\mu}.
\end{equation}
With the Gaussian overlap from Eq.~(\ref{eq:quantum6}), we have
\begin{equation}\label{eq:gU0}
        g_U(\bm\lambda,\bm\mu)=\frac1{\sqrt{\mathrm{Det}(V_U(\bm\lambda,\bm\mu))}},
\end{equation}
where
\begin{equation}
    V_U(\bm\lambda,\bm\mu)=\begin{pmatrix}-UV_{\bm\mu}U^T&I_{2m}\\I_{2m}&-V_{\bm\lambda}\end{pmatrix},
\end{equation}
where for all $\bm w\in\mathbb C^m$,
\begin{equation}
    V_{\bm w}=\begin{pmatrix}0_m&\mathrm{Diag}(\bm w)\\\mathrm{Diag}(\bm w)&0_m\end{pmatrix}.
\end{equation}
Moreover,
\begin{equation}
    \mathrm{Det}(V_U(\bm\lambda,\bm\mu))=\mathrm{Det}(I-V_{\bm\lambda}UV_{\bm\mu}U^T),
\end{equation}
so with Eq.~(\ref{eq:gU0})
\begin{equation}\label{eq:gU1}
    g_U(\bm\lambda,\bm\mu)=\frac1{\sqrt{\mathrm{Det}(I-V_{\bm\lambda}UV_{\bm\mu}U^T)}}.
\end{equation}
Setting $\bm\lambda=\bm\mu=\lambda\bm1$ for $\lambda\in\mathbb R$ and letting $\lambda$ go to $0$ in Eq.~(\ref{eq:perTMSS}), we obtain that $\mathrm{Per}(U)$ is given by
\begin{equation}
    \frac1{4^m}\!\sum_{\bm x,\bm y\in\{-1,1\}^m}x_1\dots x_my_1\dots y_m[\lambda^{2m}]g_U(\lambda\bm x,\lambda\bm y),
\end{equation}
where $[\lambda^{2m}]g_U(\lambda\bm x,\lambda\bm y)$ is the $\lambda^{2m}$ coefficient in the Taylor expansion (in $\lambda$) of
\begin{equation}
    g_U(\lambda\bm x,\lambda\bm y)=\frac1{\sqrt{\mathrm{Det}(I-\lambda^2V_{\bm x}UV_{\bm y}U^T)}}.
\end{equation}
This concludes the proof of Eq.~(\ref{eq:new1}). Once again, the relation for a generic nonzero matrix $A$ is obtained by embedding $\frac1{\|A\|}A$ as a submatrix of a unitary matrix $U$.

Now with Eq.~(\ref{eq:defgU}), for all $\bm\lambda,\bm\mu\in\mathbb C^n$,
\begin{equation}\label{eq:gU2}
    \begin{aligned}
        g_U(\bm\lambda,\bm\mu)&=\sum_{\bm p,\bm q\in\mathbb N^m}\bm\lambda^{\bm p}\bm\mu^{\bm q}\bra{\bm p\bm p}\hat U\ket{\bm q\bm q}\\
        &=\sum_{\bm p,\bm q\in\mathbb N^m}\frac{\bm\lambda^{\bm p}\bm\mu^{\bm q}}{\bm p!\bm q!}\mathrm{Per}(U_{\bm p\oplus\bm p,\bm q\oplus\bm q}),
    \end{aligned}
\end{equation}
where we used the Fock basis expansion of $\hat U$(\ref{eq:quantum3}) in the second line.
Combining Eqs.~(\ref{eq:defgU}), (\ref{eq:gU1}) and~(\ref{eq:gU2}) we obtain
\begin{equation}\label{eq:mmmttmssgen}
    \sum_{\bm p,\bm q\in\mathbb N^m}\frac{\bm\lambda^{\bm p}\bm\mu^{\bm q}}{\bm p!\bm q!}\mathrm{Per}(U_{\bm p\oplus\bm p,\bm q\oplus\bm q})=\bra{\bm\lambda^*}\hat U\ket{\bm\mu}=\frac1{\sqrt{\mathrm{Det}(I-V_{\bm\lambda}UV_{\bm\mu}U^T)}},
\end{equation}
where for all $\bm w\in\mathbb C^n$,
\begin{equation}
    V_{\bm w}=\begin{pmatrix}0_m&\mathrm{Diag}(\bm w)\\\mathrm{Diag}(\bm w)&0_m\end{pmatrix}.
\end{equation}
This concludes the proof of Eq.~(\ref{eq:new2}). Once again, the relation for a generic nonzero matrix $A$ is obtained by embedding $\frac1{\|A\|}A$ as a submatrix of a unitary matrix $U$.

When $U=A\oplus B$, with $A$ and $B$ two $m\times m$ matrices, we have $U_{\bm p\oplus\bm p,\bm q\oplus\bm q}=A_{\bm p,\bm q}\oplus B_{\bm p,\bm q}$ and the permanent of a block-diagonal matrix is the product of the permanents of the blocks, so Eq.~(\ref{eq:mmmttmssgen}) gives
\begin{equation}\label{eq:mmmttmssgenMN}
        \sum_{\bm p,\bm q\in\mathbb N^m}\frac{\bm x^{\bm p}\bm y^{\bm q}}{\bm p!\bm q!}\mathrm{Per}(A_{\bm p,\bm q})\mathrm{Per}(B_{\bm p,\bm q})=\frac1{\sqrt{\mathrm{Det}(I-V_{\bm x}(A\oplus B)V_{\bm y}(A^T\oplus B^T))}}.
\end{equation}
Writing $X=\mathrm{Diag}(x_1,\dots,x_m)$ and $Y=\mathrm{Diag}(y_1,\dots,y_m)$, we have
\begin{equation}\label{eq:detMN}
    \begin{aligned}
        \!\!\!\!\!\!\mathrm{Det}(I-V_{\bm x}(A\oplus B)V_{\bm y}(A^T\oplus B^T))&=\mathrm{Det}\!\left(I-\begin{pmatrix}0_m&X\\X&0_m\end{pmatrix}\!\begin{pmatrix}A&0_m\\0_m&B\end{pmatrix}\!\begin{pmatrix}0_m&Y\\Y&0_m\end{pmatrix}\!\begin{pmatrix}A^T&0_m\\0_m&B^T\end{pmatrix}\right)\!\!\!\!\!\!\!\!\!\!\!\!\!\!\!\!\!\!\!\!\\
        &=\mathrm{Det}\!\begin{pmatrix}I-XBYA^T&0_m\\0_m&I-XAYB^T\end{pmatrix}\\
        &=\mathrm{Det}(I-XBYA^T)\,\mathrm{Det}(I-XAYB^T)\\
        &=\mathrm{Det}(I-XAYB^T)^2,
    \end{aligned}
\end{equation}
where the last line is obtained by using Sylvester's determinant theorem $\mathrm{Det}(I+MN)=\mathrm{Det}(I+NM)$ together with the transpose for the first determinant. Combining Eqs.~(\ref{eq:mmmttmssgenMN}) and~(\ref{eq:detMN}) concludes the proof of Theorem~\ref{th:mmmt+}. 

This theorem reduces to the MacMahon master theorem when $A=I$ or $B=I$, since $\mathrm{Per}(I_{\bm p,\bm q})=\bm p!\delta_{\bm p,\bm q}$. 
Another particular case of interest is when $B=J_m$, where $J_m$ is the all-$1$ matrix of size $m\times m$, which satisfies $\mathrm{Per}(J_m)=m!$. In this case, we obtain, for all $\bm x,\bm y\in\mathbb C^m$,
\begin{equation}\label{eq:new4}
    \begin{aligned}
        \sum_{\substack{n\in\mathbb N,\bm p,\bm q\in\mathbb N^m\\|\bm p|=|\bm q|=n}}\frac{n!\bm x^{\bm p}\bm y^{\bm q}}{\bm p!\bm q!}\mathrm{Per}(A_{\bm p,\bm q})&=\frac1{\mathrm{Det}(I-AYJ_mX)}\\
        &=\frac1{\mathrm{Det}(I-A\bm y\bm x^T)}\\
        &=\frac1{1-\bm x^TA\bm y},
    \end{aligned}
\end{equation}
where we used the fact that $A\bm y\bm x^T$ is a rank-one matrix to compute the determinant. This implies Corollary~\ref{coro:mmmt+}: for all $n\in\mathbb N^*$ and all $\bm p,\bm q\in\mathbb N^m$ with $|\bm p|=|\bm q|=n$,
\begin{equation}
    \begin{aligned}
    \mathrm{Per}(A_{\bm p,\bm q})&=\frac{\bm p!\bm q!}{n!}[\bm x^{\bm p}\bm y^{\bm q}]\left(\frac1{1-\bm x^TA\bm y}\right)\\
    &=\frac{\bm p!\bm q!}{n!}[\bm x^{\bm p}\bm y^{\bm q}]\left(\bm x^TA\bm y\right)^n,
    \end{aligned}
\end{equation}
where we used the Taylor series $\frac1{1-z}=\sum_{k=0}^{+\infty}z^k$.

We note that Theorem~\ref{th:mmmt+}---which is a generalization of the MacMahon master theorem to two matrices---may also be obtained by combining the MacMahon master theorem in Eq.~(\ref{eq:mmmt}) with the Cauchy--Binet theorem in Eq.~(\ref{eq:percomposition}), applied to the matrix $AYB^T$. As it turns out, we may apply the same proof technique inductively in order to generalize the MacMahon master theorem to $N$ matrices $A^{(1)},\dots,A^{(N)}$, for $N\ge1$ and obtain Theorem~\ref{th:mmmt++}: assuming that Eq.~(\ref{eq:new6}) holds for some $N\ge2$, we have, for all $m\times m$ matrices $B^{(1)},\dots,B^{(N)}$,
\begin{equation}\label{eq:induction0}
    \begin{aligned}
        &\sum_{\bm p_1,\dots,\bm p_N\in\mathbb N^m}\prod_{k=1}^N\frac{\bm z_k^{\bm p_k}}{\bm p_k!}\mathrm{Per}(B^{(1)}_{\bm p_1,\bm p_2})\mathrm{Per}(B^{(2)}_{\bm p_2,\bm p_3})\dots\mathrm{Per}(B^{(N)}_{\bm p_N,\bm p_1})\\
        &\quad\quad\quad\quad\quad\quad=\frac1{\mathrm{Det}(I-Z_1B^{(1)}\dots Z_NB^{(N)})},
    \end{aligned}
\end{equation}
where $\bm z_k=(z_{k1},\dots,z_{km})$ are formal variables and $Z_k=\mathrm{Diag}(\bm z_k)$ for all $k\in\{1,\dots,N\}$.

Let $A^{(1)},\dots,A^{(N+1)}$ be $m\times m$ matrices, $\bm z_{N+1}=(z_{N+1,1},\dots,z_{N+1,m})$ be formal variables and $Z_{N+1}=\mathrm{Diag}(\bm z_{N+1})$.
Setting $B^{(N)}=A^{(N)}Z_{N+1}A^{(N+1)}$ and $B^{(k)}=A^{(k)}$ for all $k\in\{1,\dots,N\}$, we obtain with Eq.~(\ref{eq:induction0}):
\begin{equation}\label{eq:induction1}
    \begin{aligned}
        &\frac1{\mathrm{Det}(I-Z_1A^{(1)}\dots Z_NA^{(N)}Z_{N+1}A^{(N+1)})}\\
        &\quad=\sum_{\bm p_1,\dots,\bm p_N\in\mathbb N^m}\prod_{k=1}^N\frac{\bm z_k^{\bm p_k}}{\bm p_k!}\mathrm{Per}(A^{(1)}_{\bm p_1,\bm p_2})\dots\mathrm{Per}(A^{(N-1)}_{\bm p_{N-1},\bm p_N})\mathrm{Per}((A^{(N)}Z_{N+1}A^{(N+1)})_{\bm p_N,\bm p_1}).
    \end{aligned}
\end{equation}
On the other hand, with the Cauchy--Binet theorem from Eq.~(\ref{eq:percomposition}),
\begin{equation}\label{eq:induction2}
    \begin{aligned}
        \mathrm{Per}((A^{(N)}Z_{N+1}A^{(N+1)})_{\bm p_N,\bm p_1})&=\!\!\!\!\!\sum_{\bm p_{N+1}\in\mathbb N^m}\!\frac1{\bm p_{N+1}!}\mathrm{Per}(A^{(N)}_{\bm p_N,\bm p_{N+1}})\mathrm{Per}((Z_{N+1}A^{(N+1)})_{\bm p_{N+1},\bm p_1})\\
        &=\!\!\!\!\!\sum_{\bm p_{N+1}\in\mathbb N^m}\!\frac{\bm z_{N+1}^{\bm p_{N+1}}}{\bm p_{N+1}!}\mathrm{Per}(A^{(N)}_{\bm p_N,\bm p_{N+1}})\mathrm{Per}(A^{(N+1)}_{\bm p_{N+1},\bm p_1}),
    \end{aligned}
\end{equation}
where we used the fact that multiplying any single row of $M$ by a variable $z$ changes $\mathrm{Per}(M)$ to $z\mathrm{Per}(M)$. Combining Eqs.~(\ref{eq:induction1}) and~(\ref{eq:induction2}) completes the induction step and the proof of Theorem~\ref{th:mmmt++}.
\end{proof}

\textbf{Applications.} Recall that the MacMahon master theorem is a useful tool for deriving combinatorial identities (see section~\ref{sec:mmmt}): it expresses the permanent of an $m\times m$ matrix $A$ with rows and columns repeated in the same way as the coefficient
\begin{equation}\label{eq:perascoefmmt}
    \mathrm{Per}(A_{\bm p,\bm p})=\bm p![\bm z^{\bm p}]\left(\frac1{\mathrm{Det}(I-ZA)}\right),
\end{equation}
where $Z=\mathrm{Diag}(\bm z)$ and $\bm p\in\mathbb N^m$, while the same permanent may also be expressed as the coefficient
\begin{equation}\label{eq:perascoefmonom}
    \mathrm{Per}(A_{\bm p,\bm p})=\bm p![\bm z^{\bm p}](A\bm z)^{\bm p}.
\end{equation}
Our generalizations of the MacMahon master theorem may be used in a similar fashion to obtain simple proofs of combinatorial identities. Let us illustrate this with an example: for $n\in\mathbb N$ and $a,b\in\mathbb C$, setting
\begin{equation}
    S_n(a,b):=\sum_{k=0}^n\binom nk^2a^kb^{n-k},
\end{equation}
we aim to prove the relation
\begin{equation}
    S_n(a,b)^2=\sum_{l=0}^n\binom{2l}l\binom{n+l}{2l}(-1)^{n-l}(a-b)^{2n-2l}S_l(a^2,b^2),
\end{equation}
which for $a=b=1$ directly implies the well-known $\sum_{k=0}^n\binom nk^2=\binom{2n}n$.

\begin{proof}
By Theorem~\ref{th:mmmt+} we have
\begin{equation}\label{eq:perascoefmmt+}
    \mathrm{Per}(A_{\bm p,\bm q})\mathrm{Per}(B_{\bm p,\bm q})=\bm p!\bm q![\bm x^{\bm p}\bm y^{\bm q}]\left(\frac1{\mathrm{Det}(I-XAYB^T)}\right),
\end{equation}
where $X=\mathrm{Diag}(x_1,\dots,x_m)$ and $Y=\mathrm{Diag}(y_1,\dots,y_m)$.
Setting
\begin{equation}
    A=B=\begin{pmatrix}1&a\\1&b\end{pmatrix},
\end{equation}
for $a,b\in\mathbb C$, $n\in\mathbb N^*$, and $\bm p=\bm q=(n,n)$, we obtain with Eq.~(\ref{eq:perascoefmonom}):
\begin{equation}\label{eq:exmmt+inter}
    \begin{aligned}
        \mathrm{Per}(A_{\bm p,\bm q})&=(n!)^2[z_1^nz_2^n](z_1+az_2)^n(z_1+bz_2)^n\\
        &=(n!)^2\sum_{k+l=n}\binom nk\binom nla^kb^l\\
        &=(n!)^2S_n(a,b).
    \end{aligned}
\end{equation}
On the other hand, with Eq.~(\ref{eq:perascoefmmt+}) we have
\begin{align}
        \nonumber\mathrm{Per}(A_{\bm p,\bm q})^2&=(n!)^4[x_1^nx_2^ny_1^ny_2^n]\left(\frac1{\mathrm{Det}(I-\mathrm{Diag}(x_1,x_2)A\mathrm{Diag}(y_1,y_2)A^T}\right)\\
        \nonumber&=(n!)^4[x_1^nx_2^ny_1^ny_2^n]\left(\frac1{1-x_1y_1-a^2x_1y_2-x_2y_1-b^2x_2y_2+(a-b)^2x_1x_2y_1y_2}\right)\\
        &=(n!)^4\sum_{k=n}^{2n}[x_1^nx_2^ny_1^ny_2^n]\left(x_1y_1+a^2x_1y_2+x_2y_1+b^2x_2y_2-(a-b)^2x_1x_2y_1y_2\right)^k\displaybreak\\
        \nonumber&=(n!)^4\sum_{k=n}^{2n}\sum_{j=0}^k\binom kj(-1)^{k-j}(a-b)^{2k-2j}\!\!\!\!\!\!\!\!\sum_{n_{11}+n_{12}+n_{21}+n_{22}=j}\frac{j!a^{2n_{12}}b^{2n_{22}}}{n_{11}!n_{12}!n_{21}!n_{22}!}\!\!\!\!\!\!\\
        \nonumber&\quad\quad\times[x_1^nx_2^ny_1^ny_2^n]\left((x_1x_2y_1y_2)^{k-j}x_1^{n_{11}+n_{12}}x_2^{n_{21}+n_{22}}y_1^{n_{11}+n_{21}}y_2^{n_{12}+n_{22}}\right),\!\!\!\!\!\!
\end{align}
where we used the Taylor expansion of $\frac1{1-z}$ in the third line and the multinomial theorem in the last line. The indices in the above expression must satisfy $n_{11}=n_{22}$, $n_{12}=n_{21}$, which implies $j=2n_{11}+2n_{12}$, i.e.\ $j$ is even. Moreover, $n=k-j+n_{11}+n_{12}=k-j/2$. Relabeling $j=2l$ and $n_{11}=p$ we have $k=l+n$ and we obtain
\begin{equation}
    \begin{aligned}
        \mathrm{Per}(A_{\bm p,\bm q})^2&=(n!)^4\sum_{l=0}^n\binom{n+l}{2l}(-1)^{n-l}(a-b)^{2n-2l}\sum_{p=0}^l\frac{(2l)!a^{2l-2p}b^{2p}}{p!^2(l-p)!^2}\\
        &=(n!)^4\sum_{l=0}^n\binom{2l}l\binom{n+l}{2l}(-1)^{n-l}(a-b)^{2n-2l}S_l(a^2,b^2).
    \end{aligned}
\end{equation}
With Eq.~(\ref{eq:exmmt+inter}) this proves
\begin{equation}
    S_n(a,b)^2=\sum_{l=0}^n\binom{2l}l\binom{n+l}{2l}(-1)^{n-l}(a-b)^{2n-2l}S_l(a^2,b^2).
\end{equation}
\end{proof}

\subsection{New generating functions}

We have encountered several generating functions for the permanent of an $m\times m$ matrix $A$ with differently repeated rows and columns. Jackson's formula gives~\cite{jackson1977unification}:
\begin{equation}
    e^{\bm x^TA\bm y}=\!\!\!\!\!\sum_{\substack{n\in\mathbb N,\bm p,\bm q\in\mathbb N^m\\|\bm p|=|\bm q|=n}}\frac{\bm x^{\bm p}\bm y^{\bm q}}{\bm p!\bm q!}\mathrm{Per}(A_{\bm p,\bm q}),
\end{equation}
for formal variables $\bm x=(x_1,\dots,x_m),\bm y=(y_1,\dots,y_m)$. Moreover, we have obtained in Eq.~(\ref{eq:new4}), as a corollary of our generalization of the MacMahon master theorem:
\begin{equation}
    \frac1{1-\bm x^TA\bm y}=\!\!\!\!\!\sum_{\substack{n\in\mathbb N,\bm p,\bm q\in\mathbb N^m\\|\bm p|=|\bm q|=n}}n!\frac{\bm x^{\bm p}\bm y^{\bm q}}{\bm p!\bm q!}\mathrm{Per}(A_{\bm p,\bm q}).
\end{equation}
Furthermore, Corollary~\ref{coro:mmmt+} gives:
\begin{equation}\label{eq:quintessence}
    (\bm x^TA\bm y)^n=\!\!\!\!\!\sum_{\substack{\bm p,\bm q\in\mathbb N^m\\|\bm p|=|\bm q|=n}}n!\frac{\bm x^{\bm p}\bm y^{\bm q}}{\bm p!\bm q!}\mathrm{Per}(A_{\bm p,\bm q}).
\end{equation}
In this section we prove Theorem~\ref{th:gen}, which generalizes all three statements: for any series $f(z)=\sum_{n=0}^{+\infty}f_nz^n$,
\begin{equation}
    f(\bm x^TA\bm y)=\!\!\!\!\!\sum_{\substack{n\in\mathbb N,\bm p,\bm q\in\mathbb N^m\\|\bm p|=|\bm q|=n}}f_nn!\frac{\bm x^{\bm p}\bm y^{\bm q}}{\bm p!\bm q!}\mathrm{Per}(A_{\bm p,\bm q}).
\end{equation}
One may prove this statement by linearity, using Eq.~(\ref{eq:quintessence}). In what follows, we give a direct quantum-inspired proof which is similar to that of Jackson's formula from section~\ref{sec:pergen}. 

\begin{proof}[Proof of \cref{th:gen}]
For $\bm\alpha,\bm\beta\in\mathbb C^m$ and a unitary matrix $U$,
\begin{equation}\label{eq:proofgen1}
    \begin{aligned}
        \!\!\!\!\!\!\!\!\!\!\sum_{\substack{n\in\mathbb N,\bm p,\bm q\in\mathbb N^m\\|\bm p|=|\bm q|=n}}\!\!\!\!f_nn!\frac{\bm\alpha^{\bm p}\bm\beta^{\bm q}}{\bm p!\bm q!}\mathrm{Per}(U_{\bm p,\bm q})&=e^{\frac12\|\bm\alpha\|^2+\frac12\|\bm\beta\|^2}\sum_{n=0}^{+\infty}f_nn!\!\!\!\!\!\!\sum_{\substack{\bm p,\bm q\in\mathbb N^m\\|\bm p|=|\bm q|=n}}\langle\bm\alpha^*\ket{\bm p}\!\bra{\bm p}\hat U\ket{\bm q}\!\bra{\bm q}\bm\beta\rangle\!\!\!\!\!\!\!\!\!\!\!\!\!\!\!\!\!\!\!\\
        &=e^{\frac12\|\bm\alpha\|^2+\frac12\|\bm\beta\|^2}\sum_{n=0}^{+\infty}f_nn!\bra{\bm\alpha^*}\Pi_n\hat U\Pi_n\ket{\bm\beta}\\
        &=e^{\frac12\|\bm\alpha\|^2+\frac12\|\bm\beta\|^2}\sum_{n=0}^{+\infty}f_nn!\bra{\bm\alpha^*}\Pi_n\hat U\ket{\bm\beta}\\
        &=e^{\frac12\|\bm\alpha\|^2+\frac12\|\bm\beta\|^2}\sum_{n=0}^{+\infty}f_nn!\bra{\bm\alpha^*}\Pi_n\ket{U\bm\beta},
    \end{aligned}
\end{equation}
where we used the Fock basis expansions of coherent states (\ref{eq:quantum1}) and of $\hat U$~(\ref{eq:quantum3}) in the first line, the definition of the projector $\Pi_n=\sum_{|\bm p|=n}\ket{\bm p}\!\bra{\bm p}$ in the second line, the fact that passive linear operations conserve the total number of photons (\ref{eq:quantum4}) in the third line, and the action of $\hat U$ on coherent states (\ref{eq:quantum5}) in the fourth line.
Now,
\begin{equation}\label{eq:proofgen2}
    \begin{aligned}
        e^{\frac12\|\bm\alpha\|^2+\frac12\|\bm\beta\|^2}\bra{\bm\alpha^*}\Pi_n\ket{U\bm\beta}&=e^{\frac12\|\bm\alpha\|^2+\frac12\|\bm\beta\|^2}\sum_{|\bm p|=n}\langle\bm\alpha^*\ket{\bm p}\!\bra{\bm p}U\bm\beta\rangle\\
        &=\sum_{|\bm p|=n}e^{\frac12\|\bm\beta\|^2-\frac12\|U\bm\beta\|^2}\frac{\bm\alpha^{\bm p}(U\bm\beta)^{\bm p}}{\bm p!}\\
        &=\sum_{|\bm p|=n}\frac{\bm\alpha^{\bm p}(U\bm\beta)^{\bm p}}{\bm p!}\\
        &=\frac1{n!}(\bm\alpha^TU\bm\beta)^n,
    \end{aligned}
\end{equation}
where we used the Fock basis expansion of coherent states (\ref{eq:quantum1}) in the second line, the fact that $U$ is unitary in the third line, and the multinomial theorem in the last line. 
Combining Eqs.~(\ref{eq:proofgen1}) and~(\ref{eq:proofgen2}) completes the proof for unitary matrices. Once again, the relation for a generic nonzero matrix $A$ is obtained by embedding $\frac1{\|A\|}A$ as a submatrix of a unitary matrix $U$.

As a result, when $f_n\neq0$ we have, for any $m\times m$ matrix $A$, all $n\in\mathbb N$ and all $\bm p,\bm q\in\mathbb N^m$ such that $|\bm p|=|\bm q|=n$,
\begin{equation}
    \begin{aligned}
        \underset{\bm x,\bm y\in\mathbb T^m}{\mathbb E}\left[\frac{\bm p!\bm q!}{\bm x^{\bm p}\bm y^{\bm q}}\frac{f(\bm x^TA\bm y)}{\partial^n_zf(z)\big\vert_{z=0}}\right]&=\frac{\bm p!\bm q!}{\partial^n_zf(z)\big\vert_{z=0}}\sum_{\substack{k\in\mathbb N,\bm s,\bm t\in\mathbb N^m\\|\bm s|=|\bm t|=k}}f_kk!\frac{\mathrm{Per}(A_{\bm s,\bm t})}{\bm s!\bm t!}\\
        &\quad\times\prod_{j=1}^m\left(\int_{\theta_j,\varphi_j\in[0,2\pi]}e^{i\theta_j(s_j-p_j)}e^{i\varphi_j(t_j-q_j)}\frac{d\theta_jd\varphi_j}{4\pi^2}\right)\\
        &=\frac1{\partial^n_zf(z)\big\vert_{z=0}}f_nn!\mathrm{Per}(A_{\bm p,\bm q})\\
        &=\mathrm{Per}(A_{\bm p,\bm q}),
    \end{aligned}
\end{equation}
where we have set $\bm x=(e^{i\theta_1},\dots,e^{i\theta_m})\in\mathbb T^m$ and $\bm y=(e^{i\varphi_1},\dots,e^{i\varphi_m})\in\mathbb T^m$ in the first line, and where we have used $\int_0^{2\pi}e^{i(k-l)\theta}\frac{d\theta}{2\pi}=\delta_{kl}$ in the third line.
\end{proof}

\textbf{Applications.} These generating functions may be used to derive remarkable identities for the permanent: for instance, for $A$ and $B$ two $m\times m$ matrices, taking $f(z)=e^z$ and equating the $\bm x^{\bm p}\bm y^{\bm q}$ coefficients of $e^{\bm x^T(A+B)\bm y}$ and $e^{\bm x^TA\bm y}e^{\bm x^TB\bm y}$ yields the sum formula~\cite{minc1984permanents,percus2012combinatorial}:
\begin{equation}\label{eq:persum}
    \mathrm{Per}((A+B)_{\bm p,\bm q})=\sum_{\substack{\bm s+\bm t=\bm p\\\bm u+\bm v=\bm q}}\frac{\bm p!\bm q!}{\bm s!\bm t!\bm u!\bm v!}\mathrm{Per}(A_{\bm s,\bm u})\mathrm{Per}(B_{\bm t,\bm v}),
\end{equation}
for all $\bm p,\bm q\in\mathbb N^m$.

Similarly, for $A$ an $m\times m$ matrix, taking $f(z)=z^p$ for $p=k,l,k+l$ and equating the $\bm x^{\bm p}\bm y^{\bm q}$ coefficients of $(\bm x^TA\bm y)^{k+l}$ and $(\bm x^TA\bm y)^k(\bm x^TA\bm y)^l$ yields the Laplace expansion formula~\cite{minc1984permanents,percus2012combinatorial}:
\begin{equation}
    \mathrm{Per}(A_{\bm p,\bm q})=\frac{k!l!}{(k+l)!}\sum_{\substack{\bm s+\bm t=\bm p\\\bm u+\bm v=\bm q\\|\bm s|=|\bm u|=k\\|\bm t|=|\bm v|=l}}\frac{\bm p!\bm q!}{\bm s!\bm t!\bm u!\bm v!}\mathrm{Per}(A_{\bm s,\bm u})\mathrm{Per}(A_{\bm t,\bm v}),
\end{equation}
for all $k,l\in\mathbb N$ and all $\bm p,\bm q\in\mathbb N^m$ with $|\bm p|=|\bm q|=k+l$.

Finally, for $A$ and $B$ two $m\times m$ matrices, taking $f(z)=-\log(1-z)$ and equating the $\bm x^{\bm p}\bm y^{\bm q}$ coefficients of $-\log(1-\bm x^TA\bm y)-\log(1-\bm x^TB\bm y)$ and $-\log[(1-\bm x^TB\bm y)(1-\bm x^TA\bm y)]$ yields Theorem~\ref{th:sumper}: 
\begin{equation}
    \begin{aligned}
        &\mathrm{Per}(A_{\bm p,\bm q})+\mathrm{Per}(B_{\bm p,\bm q})\\
        &\quad=\sum_{k=0}^{\lfloor\frac n2\rfloor}\frac{(-1)^k}{\binom{n-1}k}\!\!\!\!\!\!\!\sum_{\substack{\bm a+\bm b+\bm c=\bm p\\\bm a'+\bm b'+\bm c'=\bm q\\|\bm a|=|\bm b|=|\bm a'|=|\bm b'|=k}}\!\!\!\!\!\!\!\frac{\bm p!\bm q!}{\bm a!\bm b!\bm c!\bm a'!\bm b'!\bm c'!}\mathrm{Per}(A_{\bm a,\bm a'})\mathrm{Per}(B_{\bm b,\bm b'})\mathrm{Per}((A+B)_{\bm c,\bm c'}),
    \end{aligned}
\end{equation}
after a derivation which we detail below.

\begin{proof}[Proof of \cref{th:sumper}]
Let us write
\begin{equation}
    -\log[(1-\bm x^TB\bm y)(1-\bm x^TA\bm y)]=-\log[1-\bm x^T(A+B-B\bm y\bm x^TA)\bm y)].
\end{equation}
Applying Theorem~\ref{th:gen} for $f(z)=-\log(1-z)$ with $z=\bm x^TB\bm y,\bm x^TA\bm y,\bm x^T(A+B-B\bm y\bm x^TA)\bm y$ and considering the $\bm x^{\bm p}\bm y^{\bm q}$ coefficient we obtain for all $n\in\mathbb N^*$ and all $\bm p,\bm q\in\mathbb N$ with $|\bm p|=|\bm q|=n$:
\begin{equation}\label{eq:sumperinter0}
    \begin{aligned}
        &\mathrm{Per}(A_{\bm p,\bm q})+\mathrm{Per}(B_{\bm p,\bm q})\\
        &\quad\quad=\frac{\bm p!\bm q!}{(n-1)!}[\bm x^{\bm p}\bm y^{\bm q}]\left(\sum_{\substack{l\in\mathbb N^*,\bm s,\bm t\in\mathbb N^m\\|\bm s|=|\bm t|=l}}(l-1)!\frac{\bm x^{\bm s}\bm y^{\bm t}}{\bm s!\bm t!}\mathrm{Per}((A+B-B\bm y\bm x^TA)_{\bm s,\bm t})\right)\\
        &\quad\quad=\sum_{\substack{l\in\mathbb N^*,\bm s,\bm t\in\mathbb N^m\\|\bm s|=|\bm t|=l}}\frac{(l-1)!}{(n-1)!}\frac{\bm p!\bm q!}{\bm s!\bm t!}[\bm x^{\bm p}\bm y^{\bm q}]\left(\bm x^{\bm s}\bm y^{\bm t}\mathrm{Per}((A+B-B\bm y\bm x^TA)_{\bm s,\bm t})\right)\\
        &\quad\quad=\sum_{\substack{l\in\mathbb N^*,\bm s,\bm t\in\mathbb N^m\\|\bm s|=|\bm t|=l\\\bm s\le\bm p,\bm t\le\bm q}}\frac{(l-1)!}{(n-1)!}\frac{\bm p!\bm q!}{\bm s!\bm t!}[\bm x^{\bm p-\bm s}\bm y^{\bm q-\bm t}]\left(\mathrm{Per}((A+B-B\bm y\bm x^TA)_{\bm s,\bm t})\right)\\
        &\quad\quad=\sum_{\substack{l\in\mathbb N^*,\bm s,\bm t,\bm u,\bm v\in\mathbb N^m\\|\bm s|=|\bm t|=l\\\bm s+\bm u=\bm p\\\bm t+\bm v=\bm q}}\frac{(l-1)!}{(n-1)!}\frac{\bm p!\bm q!}{\bm s!\bm t!}[\bm x^{\bm u}\bm y^{\bm v}]\left(\mathrm{Per}((A+B-B\bm y\bm x^TA)_{\bm s,\bm t})\right).
    \end{aligned}
\end{equation}
Applying the sum formula from Eq.~(\ref{eq:persum}) to the matrices $(A+B)$ and $(-B\bm y\bm x^TA)$ we obtain, for all $l\in\mathbb N^*$ and all $\bm s,\bm t\in\mathbb N$ with $|\bm s|=|\bm t|=l$,
\begin{equation}\label{eq:sumperinter1}
    \begin{aligned}
        \mathrm{Per}((A+B-B\bm y\bm x^TA)_{\bm s,\bm t})&=\!\!\!\!\!\sum_{\substack{\bm c+\bm k=\bm s\\\bm c'+\bm k'=\bm t\\|\bm c|=|\bm c'|,|\bm k|=|\bm k'|}}\!\!\!\!\!\frac{\bm s!\bm t!}{\bm c!\bm k!\bm c'!\bm k'!}\mathrm{Per}((A+B)_{\bm c,\bm c'})\mathrm{Per}((-B\bm y\bm x^TA)_{\bm k,\bm k'})\!\!\!\!\!\!\!\!\!\!\!\!\!\!\!\\
        &=\!\!\!\!\!\sum_{\substack{\bm c+\bm k=\bm s\\\bm c'+\bm k'=\bm t\\|\bm c|=|\bm c'|,|\bm k|=|\bm k'|}}\!\!\!\!\!\frac{(-1)^{|\bm k|}\bm s!\bm t!}{\bm c!\bm k!\bm c'!\bm k'!}\mathrm{Per}((A+B)_{\bm c,\bm c'})\mathrm{Per}((B\bm y\bm x^TA)_{\bm k,\bm k'}).\!\!\!\!\!\!\!\!\!\!\!\!\!\!\!
    \end{aligned}
\end{equation}
The permanent of the outer product $\bm v\bm w^T$ of two vectors $\bm v$ and $\bm w$ of the same size $k$ is easily computed as 
\begin{equation}
    \begin{aligned}
        \mathrm{Per}(\bm v\bm w^T)&=\sum_{\sigma\in\mathcal S_k}\prod_{j=1}^kv_jw_{\sigma(j)}\\
        &=\sum_{\sigma\in\mathcal S_k}v_1\dots v_kw_1\dots w_k\\
        &=k!\,v_1\dots v_kw_1\dots w_k.
    \end{aligned}
\end{equation}
The matrix $(B\bm y\bm x^TA)_{\bm k,\bm k'}$ is the outer product between the vector $(B\bm y)_{\bm k}$ (obtained from the vector $B\bm y$ by repeating $k_i$ times its $i^{th}$ entry) and the vector $(A^T\bm x)_{\bm k'}$ (obtained from the vector $A^T\bm x$ by repeating $k_i'$ times its $i^{th}$ entry) of size $|\bm k|$, so its permanent is given by
\begin{equation}
        \mathrm{Per}((B\bm y\bm x^TA)_{\bm k,\bm k'})=(|\bm k|)!(B\bm y)^{\bm k}(A^T\bm x)^{\bm k'}.
\end{equation}
With Eq.~(\ref{eq:sumperinter1}) we obtain
\begin{equation}\label{eq:sumperinter2}
        \mathrm{Per}((A+B-B\bm y\bm x^TA)_{\bm s,\bm t}=\!\!\!\!\!\!\!\!\!\!\sum_{\substack{\bm c+\bm k=\bm s\\\bm c'+\bm k'=\bm t\\|\bm c|=|\bm c'|,|\bm k|=|\bm k'|}}\!\!\!\!\!\!\!\!\!\!\!\frac{(-1)^{|\bm k|}(|\bm k|)!\bm s!\bm t!}{\bm c!\bm k!\bm c'!\bm k'!}\mathrm{Per}((A+B)_{\bm c,\bm c'})(B\bm y)^{\bm k}(A^T\bm x)^{\bm k'}.
\end{equation}
Plugging this expression into Eq.~(\ref{eq:sumperinter0}) yields
\begin{align}\label{eq:sumperinter3}
        \nonumber&\quad\mathrm{Per}(A_{\bm p,\bm q})+\mathrm{Per}(B_{\bm p,\bm q})\\
        \nonumber&=\sum_{\substack{l\in\mathbb N^*,\bm s,\bm t,\bm u,\bm v\in\mathbb N^m\\|\bm s|=|\bm t|=l\\\bm s+\bm u=\bm p\\\bm t+\bm v=\bm q}}\!\frac{(l-1)!}{(n-1)!}\frac{\bm p!\bm q!}{\bm s!\bm t!}\!\!\!\!\!\!\!\!\sum_{\substack{\bm c,\bm k,\bm c',\bm k'\in\mathbb N^m\\\bm c+\bm k=\bm s\\\bm c'+\bm k'=\bm t\\|\bm c|=|\bm c'|,|\bm k|=|\bm k'|}}\!\!\!\!\!\!\!\!\!\!\!\frac{(-1)^{|\bm k|}(|\bm k|)!\bm s!\bm t!}{\bm c!\bm k!\bm c'!\bm k'!}\mathrm{Per}((A+B)_{\bm c,\bm c'})[\bm x^{\bm u}\bm y^{\bm v}](B\bm y)^{\bm k}(A^T\bm x)^{\bm k'}\\
        \nonumber&=\sum_{\substack{l\in\mathbb N^*,\bm s,\bm t,\bm u,\bm v\in\mathbb N^m\\|\bm s|=|\bm t|=l\\\bm s+\bm u=\bm p\\\bm t+\bm v=\bm q}}\!\frac{(l-1)!}{(n-1)!}\bm p!\bm q!\!\!\!\!\!\!\!\!\!\!\sum_{\substack{\bm c,\bm k,\bm c',\bm k'\in\mathbb N^m\\\bm c+\bm k=\bm s\\\bm c'+\bm k'=\bm t\\|\bm c|=|\bm c'|,|\bm k|=|\bm k'|}}\!\!\!\!\!\!\!\frac{(-1)^{|\bm k|}(|\bm k|)!}{\bm c!\bm k!\bm c'!\bm k'!}\mathrm{Per}((A+B)_{\bm c,\bm c'})[\bm x^{\bm u}](A^T\bm x)^{\bm k'}[\bm y^{\bm v}](B\bm y)^{\bm k}\\
        &=\sum_{\substack{l\in\mathbb N^*,\bm s,\bm t,\bm u,\bm v\in\mathbb N^m\\|\bm s|=|\bm t|=l\\\bm s+\bm u=\bm p\\\bm t+\bm v=\bm q}}\!\frac{(l-1)!}{(n-1)!}\bm p!\bm q!\!\!\!\!\!\!\!\!\!\!\sum_{\substack{\bm c,\bm k,\bm c',\bm k'\in\mathbb N^m\\\bm c+\bm k=\bm s\\\bm c'+\bm k'=\bm t\\|\bm c|=|\bm c'|\\|\bm k|=|\bm k'|=|\bm u|=|\bm v|}}\!\!\!\!\!\!\!\frac{(-1)^{|\bm k|}(|\bm k|)!}{\bm c!\bm k!\bm c'!\bm k'!}\mathrm{Per}((A+B)_{\bm c,\bm c'})\frac{\mathrm{Per}(A^T_{\bm k',\bm u})}{\bm u!}\frac{\mathrm{Per}(B_{\bm k,\bm v})}{\bm v!}\displaybreak\\
        \nonumber&=\sum_{\substack{l\in\mathbb N^*,\bm s,\bm t,\bm u,\bm v\in\mathbb N^m\\|\bm s|=|\bm t|=l\\\bm s+\bm u=\bm p\\\bm t+\bm v=\bm q}}\!\frac{(l-1)!}{(n-1)!}\frac{\bm p!\bm q!}{\bm u!\bm v!}\!\!\!\!\!\!\!\!\sum_{\substack{\bm c,\bm k,\bm c',\bm k'\in\mathbb N^m\\\bm c+\bm k=\bm s\\\bm c'+\bm k'=\bm t\\|\bm c|=|\bm c'|\\|\bm k|=|\bm k'|=|\bm u|=|\bm v|}}\!\!\!\!\!\!\!\frac{(-1)^{|\bm k|}(|\bm k|)!}{\bm c!\bm k!\bm c'!\bm k'!}\mathrm{Per}(A_{\bm u,\bm k'})\mathrm{Per}(B_{\bm k,\bm v})\mathrm{Per}((A+B)_{\bm c,\bm c'})\\
        \nonumber&=\sum_{k=0}^{\lfloor\frac n2\rfloor}\frac{(-1)^kk!(n-k-1)!}{(n-1)!}\!\!\!\!\!\!\!\!\!\!\!\sum_{\substack{\bm a,\bm b,\bm c,\bm a',\bm b',\bm c'\in\mathbb N^m\\\bm a+\bm b+\bm c=\bm p\\\bm a'+\bm b'+\bm c'=\bm q\\|\bm a|=|\bm a'|=|\bm b|=|\bm b'|=k}}\!\!\frac{\bm p!\bm q!}{\bm a!\bm b!\bm c!\bm a'!\bm b'!\bm c'!}\mathrm{Per}(A_{\bm a,\bm a'})\mathrm{Per}(B_{\bm b,\bm b'})\mathrm{Per}((A+B)_{\bm c,\bm c'}),
\end{align}
where we used Eq.~(\ref{eq:permonom}) in the third line, $\mathrm{Per}(M^T)=\mathrm{Per}(M)$ in the fourth line, and where we relabeled $\bm u$ as $\bm a$, $\bm k'$ as $\bm a'$, $\bm k$ as $\bm b$, $\bm v$ as $\bm b'$, and $|\bm k|=|\bm k'|=|\bm u|=|\bm v|$ as $k=n-l$ in the last line. The cutoff of the sum at $\lfloor\frac n2\rfloor$ comes from the fact that $n=|\bm p|=|\bm a|+|\bm b|+|\bm c|=2k+|\bm c|\ge2k$.
Writing $\frac{(-1)^kk!(n-k-1)!}{(n-1)!}=\frac{(-1)^k}{\binom{n-1}k}$ completes the proof of Theorem~\ref{th:sumper}.
\end{proof}

\section{Boson Sampling with input cat states}
\label{sec:BScat}

In this section, we discuss the classical complexity of sampling from the output distribution of linear optical computations with input cat states, which we refer to as Boson Sampling with input cat states.

Since the introduction of Boson Sampling by Aaronson and Arkhipov~\cite{Aaronson2013} for the demonstration of quantum computational advantage using noninteracting bosons, several variants of this model have been analyzed~\cite{lund2014boson,olson2015sampling,hamilton2017gaussian,lund2017exact,BShomodyne2017,Chabaud2017hom,quesada2018gaussian,deshpande2022quantum}. These variants were introduced to address two different challenges: on the one hand, to reduce the experimental burden associated with the demonstration of quantum advantage; on the other hand, to understand the resources for this quantum advantage and which computational models are able to reproduce it. Boson Sampling with input cat states is another such variant and has been first considered in~\cite{rohde2015evidence}, where its classical hardness was shown for exact sampling and argued for approximate sampling. We strengthen these results in the following.

\medskip

We first give a proof of Lemma~\ref{lem:BScat}, which provides a closed form expression for output amplitudes of Boson Sampling with input cat states.

\begin{proof}[Proof of Lemma~\ref{lem:BScat}]
In the quantum-inspired proof of Glynn's formula in section~\ref{sec:Glynn} we have obtained the following identity:
\begin{equation}\label{eq:recallcat}
    \langle p_1\dots p_m|\hat U(\ket{\tilde{\mathrm{cat}}_\alpha}^{\otimes n}\otimes\ket 0^{\otimes m-n})=\frac{\alpha^{|\bm p|-n}e^{-\frac n2|\alpha|^2}}{2^n\sqrt{\bm p!}}\sum_{\bm x\in\{-1,1\}^n}x_1\dots x_n\prod_{i=1}^m\left(\sum_{j=1}^nu_{ij}x_j\right)^{\!p_i}\!\!\!\!\!,
\end{equation}
for all $\alpha\in\mathbb C$, all $\bm p=(p_1,\dots,p_m)\in\mathbb N^m$, all $n\le m$ and any passive linear operation $\hat U$ over $m$ modes with unitary matrix $U=(u_{ij})_{1\le i,j\le m}$, where $\ket{\tilde{\mathrm{cat}}_\alpha}=\frac1{2\alpha}(\ket\alpha-\ket{-\alpha})$ is an unnormalized cat state. By definition of the cat states (\ref{eq:cat}), these states are related to normalized cat states $\ket{\mathrm{cat}_\alpha}$ by
\begin{equation}
    \ket{\mathrm{cat}_\alpha}=\frac{\alpha e^{\frac12|\alpha|^2}}{\sqrt{\sinh(|\alpha|^2)}}\ket{\tilde{\mathrm{cat}}_\alpha}.
\end{equation}
Setting $|\bm p|=n$ and replacing unnormalized cat states by normalized ones in Eq.~(\ref{eq:recallcat}), we obtain Lemma~\ref{lem:BScat}:
\begin{equation}\label{eq:BScat}
    \begin{aligned}
        \!\!\!\!\!\!\!\!\!\!\langle p_1\dots p_m|\hat U(\ket{\mathrm{cat}_\alpha}^{\otimes n}\otimes\ket 0^{\otimes m-n})&=\frac{\alpha^n}{\sqrt{\sinh^n(|\alpha|^2)}2^n\sqrt{\bm p!}}\!\sum_{\bm x\in\{-1,1\}^n}\!\!\!\!\!x_1\dots x_n\prod_{i=1}^m\!\left(\sum_{j=1}^nu_{ij}x_j\right)^{\!p_i}\!\!\!\!\!\!\!\!\!\!\!\!\!\!\!\!\!\!\!\!\!\!\!\!\!\\
        &=\frac{\alpha^n}{\sqrt{\sinh^n(|\alpha|^2)}}\frac{\mathrm{Per}(U_{\bm p,\bm1\oplus\bm0})}{\sqrt{\bm p!}}\\
        &=\frac{\alpha^n}{\sqrt{\sinh^n(|\alpha|^2)}}\langle p_1\dots p_m|\hat U(\ket1^{\otimes n}\otimes\ket 0^{\otimes m-n}),
    \end{aligned}
\end{equation}
where we used  Eq.~(\ref{eq:glynnrepeated}) in the second line and the Fock basis expansion of $\hat U$ (\ref{eq:quantum3}) in the last line.
\end{proof}

Interestingly, Lemma~\ref{lem:BScat} implies that, up to a global factor, $n$ cat states of amplitude $\alpha$ reproduce exactly the Boson Sampling statistics of $n$ single photons (however, detection events $\bm p$ with $|\bm p|>n$ can occur for input cat states, contrary to single-photons). 

This result has two consequences, summarized by Theorem~\ref{th:hardBScat}. 
Firstly, it implies that Boson Sampling with input cat states is hard to sample exactly for $m\ge2n$ and for all choices of nonzero cat state amplitude unless the polynomial hierarchy collapses to its third level, since some of its outcome probabilities are $\#\textsf P$-hard to estimate multiplicatively~\cite{Aaronson2013}. Alternative proofs of this statement based either on universality under post-selection or on rejection sampling can be found in~\cite{rohde2015evidence}. 
Secondly, as we show hereafter, it implies that Boson Sampling with input cat states of small enough (nonzero) amplitudes is also hard to sample approximately in the same regime as Boson Sampling.

\begin{proof}[Proof of Theorem~\ref{th:hardBScat}]
Our proof uses arguments similar to the ones used in~\cite{rohde2015evidence}, extended to the case of approximate sampling. 

At a high level, we show how to convert any classical algorithm for approximate sampling from the output probability distribution of an $m$-mode Boson Sampling computation with $n$ input cat states to a classical algorithm for approximate sampling from the output probability distribution of an $m$-mode Boson Sampling computation with $n$ input Fock states (Lemma~\ref{lem:rejsamp}). We do so by performing rejection sampling: we run the first algorithm and only keep the samples with total photon number equal to $n$. In particular, we show that if the fraction of samples with correct photon number $n$ is large enough (namely, at least inverse polynomial in $m$), the rejection sampling subroutine is efficient and the final classical probability distribution is close to the ideal distribution of Boson Sampling with Fock state input. Then, we determine the regime of input cat state amplitude such that the fraction of samples with correct photon number $n$ is at least inverse polynomial in $m$ (Lemma~\ref{lem:amplitudecat}).

Let $P_\text{cat}(\bm p|n,\alpha):=|\langle p_1\dots p_m|\hat U(\ket{\mathrm{cat}_\alpha}^{\otimes n}\otimes\ket 0^{\otimes m-n})|^2$ be the probability of detecting $\bm p=(p_1,\dots,p_m)$ output photons in a Boson Sampling experiment with interferometer $\hat U$ and input cat states  $\ket{\mathrm{cat}_\alpha}^{\otimes n}\otimes\ket 0^{\otimes m-n}$, and let $P_\text{BS}(\bm p|n):=|\langle p_1\dots p_m|\hat U(\ket1^{\otimes n}\otimes\ket 0^{\otimes m-n})|^2$ be the probability of detecting $\bm p=(p_1,\dots,p_m)$ output photons in a Boson Sampling experiment with interferometer $\hat U$ and input single photons $\ket1^{\otimes n}\otimes\ket 0^{\otimes m-n}$.

We show the following reduction: as long as the fraction of samples with photon number $n$ is large enough, then any efficient classical algorithm for approximate sampling from $P_\text{cat}$ can be converted to an efficient classical algorithm for approximate sampling from $P_\text{BS}$.

\begin{lemma}\label{lem:rejsamp}
Suppose there exists an efficient classical algorithm $C$ for approximate sampling from $P_\text{cat}$, i.e., that takes as input the description of the $m$-mode boson sampler with $n$ input cat states with output probability distribution $P_\text{cat}$ and an error bound $\epsilon$, and that samples from a distribution $P_C$ such that $\|P_C-P_\text{cat}\|_\text{TV}\le\epsilon$ in $\text{poly }(m,\frac1\epsilon)$ time, where $\|\cdot\|_\text{TV}$ is the total variation distance.

Assume further that the fraction of samples from $P_\text{cat}$ with photon number equal to $n$ is at least inverse polynomial in $m$, i.e.\ $\sum_{|\bm p|=n}P_\text{cat}(\bm p|n,\alpha)\ge\frac1{\text{poly }m}$.

Then, there exists an efficient classical algorithm $\tilde C$ for approximate sampling from $P_\text{BS}$, i.e., a classical algorithm that takes as input the description of the $m$-mode boson sampler with $n$ input Fock states with output probability distribution $P_\text{BS}$ and an error bound $\epsilon$, and that samples from a distribution $P_C^\text{rej}$ such that $\|P_C^\text{rej}-P_\text{BS}\|_\text{TV}\le\epsilon$ in $\text{poly }(m,\frac1\epsilon)$ time.
\end{lemma}

\begin{proof}[Proof of Lemma~\ref{lem:rejsamp}] Given a sampling algorithm from a probability distribution $P$ over $\mathbb N^m$, we define the following rejection sampling subroutine: sample $\bm p$ from $P$ and compute $|\bm p|$; discard the sample if $|\bm p|\neq n$; otherwise, output $\bm p$. We denote by $P^\text{rej}$ the corresponding output probability distribution. For all $\bm p\in\mathbb N^m$, we have
\begin{equation}\label{eq:subrej}
    P^\text{rej}(\bm p)=\frac{\delta_{|\bm p|,n}P(\bm p)}{\sum_{|\bm q|=n}P(\bm q)},
\end{equation}
where $\delta$ is the Kronecker symbol.
Applying this subroutine to $P_\text{cat}$ gives a new probability distribution
\begin{equation}\label{eq:Pcatrej}
    P_\text{cat}^\text{rej}(\bm p|n,\alpha)=\frac{\delta_{|\bm p|,n}P_\text{cat}(\bm p|n,\alpha)}{\sum_{|\bm q|=n}P_\text{cat}(\bm q|n,\alpha)}.
\end{equation}
By Lemma~\ref{lem:BScat} (see Eq.~(\ref{eq:BScat})) we have, for all $\bm p$ such that $|\bm p|=n$,
\begin{equation}\label{eq:correspBScat}
    P_\text{cat}(\bm p|n,\alpha)=\frac{|\alpha|^{2n}}{\sinh^n(|\alpha|^2)}P_\text{BS}(\bm p|n).
\end{equation}
Hence

\begin{equation}\label{eq:fractioncat}
    \begin{aligned}
        \sum_{|\bm p|=n}P_\text{cat}(\bm p|n,\alpha)&=\frac{|\alpha|^{2n}}{\sinh^n(|\alpha|^2)}\sum_{|\bm p|=n}P_\text{BS}(\bm p|n)\\
        &=\frac{|\alpha|^{2n}}{\sinh^n(|\alpha|^2)},
    \end{aligned}
\end{equation}
where we used $\sum_{|\bm p|=n}P_\text{BS}(\bm p|n)=1$ in the second line since $\hat U$ conserves the total number of photons~(\ref{eq:quantum4}). Combining Eqs.~(\ref{eq:Pcatrej}-\ref{eq:fractioncat}), we thus obtain
\begin{equation}\label{eq:catrejeqBS}
    P_\text{BS}(\bm p|n)=P_\text{cat}^\text{rej}(\bm p|n,\alpha).
\end{equation}
This reproduces the argument for exact sampling hardness from~\cite{rohde2015evidence}.

To show approximate hardness, we define the classical approximate sampling algorithm $\tilde C$ by running the algorithm $C$ with input error bound $\frac\epsilon2\sum_{|\bm p|=n}P_\text{cat}(\bm p|n,\alpha)$ and applying the rejection sampling subroutine. By assumption, running the algorithm $C$ can be done in time
\begin{equation}
    \text{poly}\left(m,\frac2{\epsilon\sum_{|\bm p|=n}P_\text{cat}(\bm p|n,\alpha)}\right)=\text{poly}\left(m,\frac1\epsilon\right),
\end{equation}
since $\sum_{|\bm p|=n}P_\text{cat}(\bm p|n,\alpha)\ge\frac1{\text{poly }m}$, also by assumption. On the other hand, the rejection sampling subroutine induces a computational overhead scaling with the inverse of the fraction of kept samples $\sum_{|\bm q|=n}P_C(\bm q)$. We have
\begin{equation}
    \begin{aligned}
        \sum_{|\bm q|=n}P_\text{cat}(\bm q|n,\alpha)-\sum_{|\bm q|=n}P_C(\bm q)&\le\left|\sum_{|\bm q|=n}P_C(\bm q)-\sum_{|\bm q|=n}P_\text{cat}(\bm q|n,\alpha)\right|\\
        &\le\sum_{|\bm q|=n}\left|P_C(\bm q)-P_\text{cat}(\bm q|n,\alpha)\right|\\
        &\le\sum_{\bm q\in\mathbb N^m}\left|P_C(\bm q)-P_\text{cat}(\bm q|n,\alpha)\right|\\
        &=2\|P_C-P_\text{cat}\|_\text{TV}\\
        &\le\epsilon\sum_{|\bm p|=n}P_\text{cat}(\bm p|n,\alpha),
    \end{aligned}
\end{equation}
where we used the triangle inequality in the second line, the definition of the total variation distance in the fourth line, and the definition of $C$ in the last line. Hence,
\begin{equation}
    \begin{aligned}
        \sum_{|\bm q|=n}P_C(\bm q)&\ge(1-\epsilon)\sum_{|\bm q|=n}P_\text{cat}(\bm q|n,\alpha)\\
        &\ge\frac{1-\epsilon}{\text{poly }m},
    \end{aligned}
\end{equation}
so that the computational overhead induced by the rejection sampling subroutine scales as $\text{poly}(m,\frac1\epsilon)$.

Overall, the classical algorithm $\tilde C$ outputs samples from the probability distribution $P_C^\text{rej}$ in $\text{poly}(m,\frac1\epsilon)$ time (with probability exponentially close to $1$). Moreover,
\begin{align}
        \nonumber\|P_C^\text{rej}-P_\text{BS}\|_\text{TV}&=\frac12\sum_{\bm p\in\mathbb N^m}|P_C^\text{rej}(\bm p)-P_\text{BS}(\bm p|n)|\\
        \nonumber&=\frac12\sum_{|\bm p|=n}|P_C^\text{rej}(\bm p)-P_\text{BS}(\bm p|n)|\\
        \nonumber&=\frac12\sum_{|\bm p|=n}\left|P_C^\text{rej}(\bm p)-P_\text{cat}^\text{rej}(\bm p|n,\alpha)\right|\\
        \nonumber&=\frac12\sum_{|\bm p|=n}\left|\frac{P_C(\bm p)}{\sum_{|\bm q|=n}P_C(\bm q)}-\frac{P_\text{cat}(\bm p|n,\alpha)}{\sum_{|\bm q|=n}P_\text{cat}(\bm q|n,\alpha)}\right|\\
        &\le\frac12\sum_{|\bm p|=n}\Bigg[\frac{\left|P_C(\bm p)-P_\text{cat}(\bm p|n,\alpha)\right|}{\sum_{|\bm q|=n}P_\text{cat}(\bm q|n,\alpha)}\\
        \nonumber&\quad\quad\quad\quad\quad\quad\quad\quad\quad+P_C(\bm p)\left|\frac1{\sum_{|\bm q|=n}P_C(\bm q)}-\frac1{\sum_{|\bm q|=n}P_\text{cat}(\bm q|n,\alpha)}\right|\Bigg]\\
        \nonumber&\le\frac1{2\sum_{|\bm q|=n}P_\text{cat}(\bm q|n,\alpha)}\sum_{|\bm p|=n}\Bigg[\left|P_C(\bm p)-P_\text{cat}(\bm p|n,\alpha)\right|\\
        \nonumber&\quad\quad\quad\quad\quad\quad\quad\quad\quad+\frac{P_C(\bm p)}{\sum_{|\bm q|=n}P_C(\bm q)}\left|\sum_{|\bm q|=n}P_\text{cat}(\bm q|n,\alpha)-\sum_{|\bm q|=n}P_C(\bm q)\right|\Bigg]\displaybreak\\
        \nonumber&\le\frac1{\sum_{|\bm q|=n}P_\text{cat}(\bm q|n,\alpha)}\Bigg[\frac12\sum_{\bm p\in\mathbb N^m}\left|P_C(\bm p)-P_\text{cat}(\bm p|n,\alpha)\right|\\
        \nonumber&\quad\quad\quad\quad\quad\quad\quad\quad\quad+\sum_{|\bm p|=n}\frac{P_C(\bm p)}{\sum_{|\bm q|=n}P_C(\bm q)}\frac12\sum_{\bm q\in\mathbb N^m}\left|P_\text{cat}(\bm q|n,\alpha)-P_C(\bm q)\right|\Bigg]\\
        \nonumber&=\frac{2\|P_C-P_\text{cat}\|_\text{TV}}{\sum_{|\bm q|=n}P_\text{cat}(\bm q|n,\alpha)}\\
        \nonumber&\le\epsilon.
\end{align}
where we used the definition of the total variation distance in the first and eighth lines, the fact that $P_C^\text{rej}$ and $P_\text{BS}$ are supported on samples with total photon number $n$ in the second line, Eq.~(\ref{eq:catrejeqBS}) in the third line, Eq.~(\ref{eq:subrej}) in the fourth line and the triangle inequality in the fifth line.

Hence, $\tilde C$ is an efficient classical algorithm for approximate sampling from $P_\text{BS}$.
\end{proof}

\noindent We are left with identifying the regime of input cat state amplitude $\alpha\in\mathbb C$ such that the assumption in Lemma~\ref{lem:rejsamp} is satisfied, i.e.\ $\sum_{|\bm p|=n}P_\text{cat}(\bm p|n,\alpha)\ge\frac1{\text{poly }m}$.

\begin{lemma}\label{lem:amplitudecat}
For $0<|\alpha|=\mathcal O(n^{-1/4}\log^{1/4}m)$ and $m=\text{poly }n$, the detection events satisfying $|\bm p|=n$ represent at least an inverse-polynomial fraction of the outcomes for Boson Sampling with input cat states of amplitude $\alpha$, i.e.\ $\sum_{|\bm p|=n}P_\text{cat}(\bm p|n,\alpha)\ge\frac1{\text{poly }m}$.
\end{lemma}

\begin{proof}[Proof of Lemma~\ref{lem:amplitudecat}]
Let $\alpha\in\mathbb C$ such that $0<|\alpha|=\mathcal O(n^{-1/4}\log^{1/4}m)$. The fraction of outcomes $\bm p\in\mathbb N^m$ for $m$-mode Boson Sampling with $n$ input cat states of amplitude $\alpha$ is given by
\begin{equation}
    \begin{aligned}
        \sum_{|\bm p|=n}P_\text{cat}(\bm p|n,\alpha)&=\frac{|\alpha|^{2n}}{\sinh^n(|\alpha|^2)}\\
        &=e^{n\log(|\alpha|^2)-n\log(\sinh|\alpha|^2)}\\
        &=e^{n\log(|\alpha|^2)-n\log(|\alpha|^2+\frac16|\alpha|^6+O(|\alpha|^{10}))}\\
        &=e^{-n\log(1+\frac16|\alpha|^4+O(|\alpha|^8))}\\
        &=e^{-\frac16n|\alpha|^4+O(n|\alpha|^8)}\\
        &=e^{-O(\log m)+O(n^{-2}\log^2m)}\\
        &\ge\frac1{\text{poly }m},
    \end{aligned}
\end{equation}
where we used Eq.~(\ref{eq:fractioncat}) in the first line, the scaling of $|\alpha|$ in the sixth line, and $O(n^{-2}\log^2m)=O(1)$ in the last line.
\end{proof}

\noindent Combining Lemma~\ref{lem:rejsamp} and Lemma~\ref{lem:amplitudecat}, using any efficient classical algorithm for approximate sampling from the output probability distribution of a Boson Sampling instance with input cat states with $0<|\alpha|=\mathcal O(n^{-1/4}\log^{1/4}n)$, one may also efficiently sample approximately from the output probability distribution of the corresponding Boson Sampling instance with input single photons efficiently, by keeping only the samples $\bm p$ satisfying $|\bm p|=n$. 

This proves that Boson Sampling with input cat states with $0<|\alpha|=\mathcal O(n^{-1/4}\log^{1/4}n)$ is hard to sample approximately in the same regime as Boson Sampling with input Fock states, assuming the same conjectures as in~\cite{Aaronson2013}. 
\end{proof}

\section*{Acknowledgments}

U.~C.~thanks Atul Singh Arora and Pierre-Emmanuel Emeriau for interesting discussions.
U.~C, A.~D, and S.~M.~acknowledge funding provided by the Institute for Quantum Information and Matter, an NSF Physics Frontiers Center (NSF Grant PHY-1733907).
A.~D.~also acknowledges support from the National Science Foundation RAISE-TAQS 1839204 and Amazon Web Services, AWS Quantum Program.

\bibliographystyle{linksen}
\bibliography{biblio}

\end{document}